\documentclass[a4paper,11pt]{article}
\usepackage[top=2cm,bottom=2cm,left=2cm,right=2cm]{geometry}
\usepackage[english]{babel}
\usepackage[latin1]{inputenc}
\usepackage{csquotes}
\usepackage{amsmath,amssymb}
\usepackage{graphicx}
\usepackage[normalem]{ulem}
\usepackage{url}
\usepackage{amsthm}
\usepackage{subfigure}
\usepackage{pgf}
\usepackage{algorithm}
\usepackage[noend]{algpseudocode}
\usepackage{tabularx}
\usepackage{multirow}
\usepackage{setspace}
\usepackage{textcomp}
\usepackage{indentfirst}
\usepackage{authblk}
\usepackage[font={small}]{caption, subfig}
\usepackage[round]{natbib}
\newcounter{thmc}
\newtheorem{thm}[thmc]{Theorem}

\newcounter{defnc}
\newtheorem{defn}[defnc]{Definition}

\newcolumntype{L}[1]{>{\raggedright\let\newline\\\arraybackslash\hspace{0pt}}m{#1}}
\newcolumntype{C}[1]{>{\centering\let\newline\\\arraybackslash\hspace{0pt}}m{#1}}
\newcolumntype{R}[1]{>{\raggedleft\let\newline\\\arraybackslash\hspace{0pt}}m{#1}}
\newcommand{\revisado}[1]{#1}
\newcommand{\problema}[1]{#1}
\newcommand{\riscado}[1]{}
\title{Distance geometry approach for special graph coloring problems}

\author[a]{Rosiane de Freitas}
\author[a]{Bruno Dias}
\author[b]{Nelson Maculan}
\author[b]{Jayme Szwarcfiter}
\affil[a]{\small Instituto de Computação, Universidade Federal do Amazonas, Av. Rodrigo Otávio 3000, 69000-000, Manaus, Brazil}
\affil[b]{\small COPPE, Universidade Federal do Rio de Janeiro, C.P. 68530, 21945-970, Rio de Janeiro, Brazil}
\affil[ ]{{\small }\it E-mail: rosiane@icomp.ufam.edu.br [deFreitas]; bruno.dias@icomp.ufam.edu.br [Dias]; maculan@cos.ufrj.br [Maculan]; jayme@cos.ufrj.br [Szwarcfiter]}
	
\date{\vspace*{-7mm}}

\begin{document}

\maketitle

%\author{
%\name{Segundo Autor}
%\institute{Afilia\c c\~ao}
%\iaddress{Endere\c co da Institui\c c\~ao}
%\email {e-mail}
%}

%\vspace{8mm}

\begin{abstract}
One of the most important combinatorial optimization problems is graph coloring. There are several variations of this problem involving additional constraints either on vertices or edges. They constitute models for real applications, such as channel assignment in mobile wireless networks. In this work, we consider some coloring problems involving distance constraints as weighted edges, modeling them as distance geometry problems. Thus, the vertices of the graph are considered as embedded on the real line and the coloring is treated as an assignment of positive integers to the vertices, while the distances correspond to line segments, where the goal is to find a feasible intersection of them. We formulate different such coloring problems and show feasibility conditions for some problems. We also propose implicit enumeration methods for some of the optimization problems based on branch-and-prune methods proposed for distance geometry problems in the literature. An empirical analysis was undertaken, considering equality and inequality constraints, uniform and arbitrary set of distances, and the performance of each variant of the method considering the handling and propagation of the set of distances involved.\\

\end{abstract}

%\bigskip
%\begin{keywords}
%Algorithms, graph coloring, branch-prune-and-bound, distance geometry, integer programming, telecommunications.
%
%\bigskip
%\noindent{Main Areas: OC - Combinatorial Optimization, TAG - Theory/Algorithms in Graphs.}
%\end{keywords}

%\newpage

\section{Introduction}

Let $G = (V, E)$ be an undirected graph. A \textit{$k$-coloring} of $G$ is
an assignment of colors $\{1, 2, \dots, k\}$ to the vertices of $G$ so that no
two adjacent vertices share the same color. The \textit{chromatic
number} $\chi_G$ of a graph is the minimum value of $k$ for which $G$ is \textit{$k$-colorable}. The classic graph coloring problem, which consists in finding the chromatic number of a graph, is one of the most important combinatorial optimization problems and it is known to be NP-hard \citep{karp:1972}.

There are several versions of this classic vertex coloring problem, involving
additional constraints, in both edges as vertices of the graph, with a number of practical applications as well as theoretical challenges. One of the main applications of such problems consists of assigning channels to transmitters in a mobile wireless network. Each transmitter is responsible for the calls made in the area it covers and the communication among devices is made through a channel consisting of a discrete slice of the electromagnetic spectrum. However, the channels cannot be assigned to calls in an arbitrary way, since there is the problem of interference among devices located near each other using approximate channels. There are three main types of
interferences: \textit{co-channel}, among calls of two transmitters using the same channels; \textit{adjacent channel}, among calls of two transmitters using adjacent channels and \textit{co-site}, among calls on the same cell that do not respect a minimal separation. It is necessary to assign channels to the calls such that interference is avoided and the spectrum usage is minimized \citep{audhya:2011,koster:2010,koster:1999}.

Thus, the channel assignment scenario is modeled as a graph coloring problem by
considering each transmitter as a vertex in a undirected graph and the channels to be assigned
as the colors that the vertices will receive. Some more general graph coloring problems
were proposed in the literature in order to take the separation among channels into account, such as the
T-coloring problem, also known as the Generalized Coloring Problem (GCP)
where, for each edge, the absolute difference between colors assigned to each vertex must not be
in a given forbidden set \citep{hale:1980}. The Bandwidth Coloring Problem, a special case of T-coloring where
the absolute difference between colors assigned to each vertex must be greater or equal
a certain value \citep{malaguti:2010}, and the
coloring problem with restrictions of adjacent colors (COLRAC), where there is
a restriction graph for which adjacent colors in it cannot be assigned to adjacent vertices \citep{akihiro:2002}.

The separation among channels is a type of distance constraint, so we can see the channel
assignment as a type of distance geometry (DG) problem \citep{liberti:2014} since we have to place the channels
in the transmitters respecting some distances imposed in the edges, as can be seen in
Figure \ref{fig:exNetDist}. One method to solve
DG problems is the branch-and-prune approach proposed by \citet{lavor:2012:1,lavor:2012:2}, where a solution is
built and if at some point a distance constraint is violated, then we stop this construction (prune) and try another option for the current solution in the search space. See also: \cite{mucherino:2013,lavor:2012:1,freitas:2014,freitas:2014:1,dias:2014,dias:2013:1,dias:2012}.

For graph theoretic concepts and terminology, see the book by \citet{bondy:2008}.

The main contribution of this paper consists of a distance geometry approach for special cases of T-coloring problems with distance constraints, involving a study of graph classes for which some of these distance coloring problems are unfeasible, and branch-prune-and-bound algorithms, combining concepts from the branch-and-bound method and constraint propagation, for the considered problems.

The remainder of this paper is organized as follows. Section \ref{sec:probStat} defines the distance geometry models for some special graph coloring problems.
%\riscado{Section \ref{sec:ip-cp} gives a mathematical formulation in constraint programming that better represents the theoretical distance geometry models proposed, and also gives an integer programming formulation for comparison.}
Section \ref{sec:distColProp} shows some properties regarding the structure of those distance geometry graph coloring problems, including the determination of feasibility for some graphs classes. Section \ref{sec:algs} formulates the branch-prune-and-bound (BPB) algorithms proposed for the problems \revisado{and shows properties regarding
optimality results}.
Section \ref{sec:compExp} shows results of some
experiments done with the BPB algorithms %\riscado{and the CPLEX CP Optimizer solver} 
\revisado{using randomly generated graphs for each proposed model}. Finally, Section \ref{sec:conc} concludes the paper and states the next 
steps for ongoing research.

\begin{figure}
\centering
\includegraphics[scale=0.20]{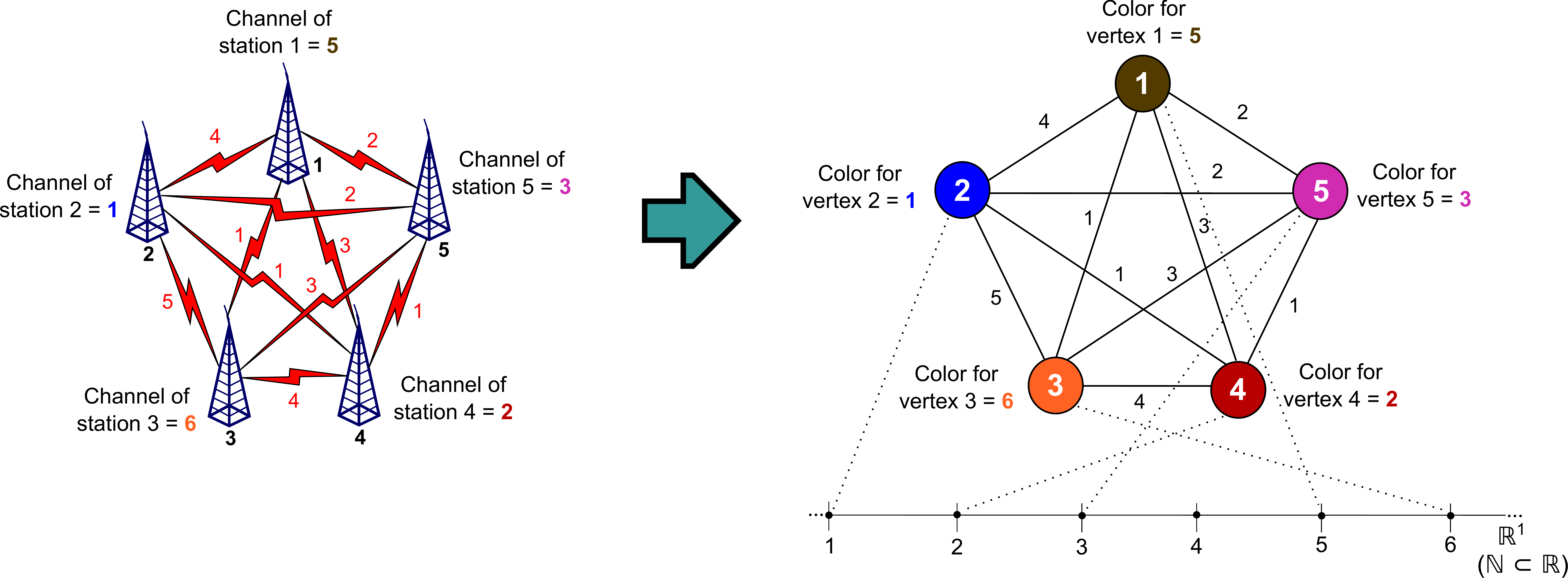}
\caption{Example of channel assignment with distance constraints, where
the separation is given by the weight in each edge. The
image on the right shows the network as an undirected graph and the projection of vertices on the real number line,
but considering only natural numbers.}
\label{fig:exNetDist}
\end{figure}

\section{Distance geometry and graph colorings}
\label{sec:probStat}

\revisado{We propose} an approach in distance geometry for special vertex coloring problems with distance 
constraints, based on the Discretizable Molecular Distance Geometry Problem (DMDGP), which is a special case of the 
Molecular Distance Geometry Problem, where the set $V$ of vertices from the input graph $G$ are ordered such that 
\revisado{the set $E$ of edges contain all cliques on quadruplets of consecutive vertices, that is, any four 
consecutive vertices induce a complete graph} ($\forall i \in \{4, \dots, n\}\ \forall j, k \in \{i-3, \dots, i\}\ 
(\{j, k\} \in E)$) \citep{lavor:2012:1}. Furthermore, a strict triangular inequality holds 
\revisado{on weights of edges between consecutive vertices in such ordering} ($\forall i \in \{2, \dots, n-1\}\ 
d_{i-1, i+1} < d_{i-1, i}\ +\ d_{i, i+1}$). All coordinates are given in $\mathbb{R}^3$ space. The position for a 
point $i$ (where $i \ge 4$) can be determined using the positions of the previous three points $i-1, i-2$ and $i-3$ 
by intersecting three spheres with radii $d_{i-3,i}, d_{i-2, i}$ and $d_{i-1, i}$, obtaining two possible points that 
are checked for feasibility.

A similar reasoning can be used in vertex coloring problems with distance constraints, where the distances that must
be respected involve the absolute difference between two values $x(i)$ and $x(j)$ (respectively, the color points 
attributed to $i$ and $j$), but for these problems the space considered is actually unidimensional.
The positioning of a vertex $i$ can be determined by using a neighbor $j$ that is already positioned. Thus, we have a 
\textit{0-sphere}, consisting of a projection of a 1-sphere (a circle), which itself is a projection of a 2-sphere 
(the three-dimensional sphere), as shown in Figure \ref{fig:spheres}.
The 0-sphere is a line segment with a radius $d_{i,j}$, and feasible colorings consist of treating the intersections 
of these 0-spheres. Figure \ref{fig:colorLines} exemplifies the correlations between these types of spheres and 
shows the example from Figure \ref{fig:exNetDist} as the positioning of these line segments.

\begin{figure}
\centering
\includegraphics[scale=0.075]{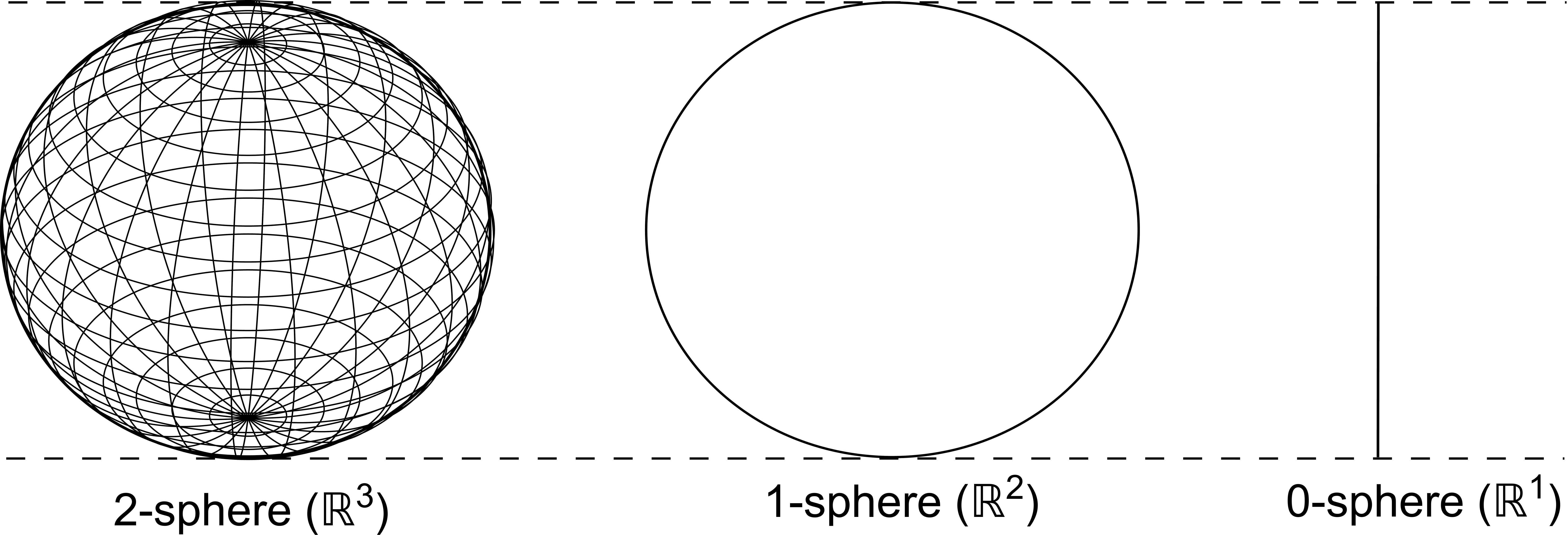}
\caption{Some types of $n$-spheres. A ($n-1$)-sphere is a projection of a $n$-sphere on a lower dimension.}
\label{fig:spheres}
\vspace{1.0cm}
\end{figure}

\begin{figure}
\centering
\includegraphics[scale=0.18]{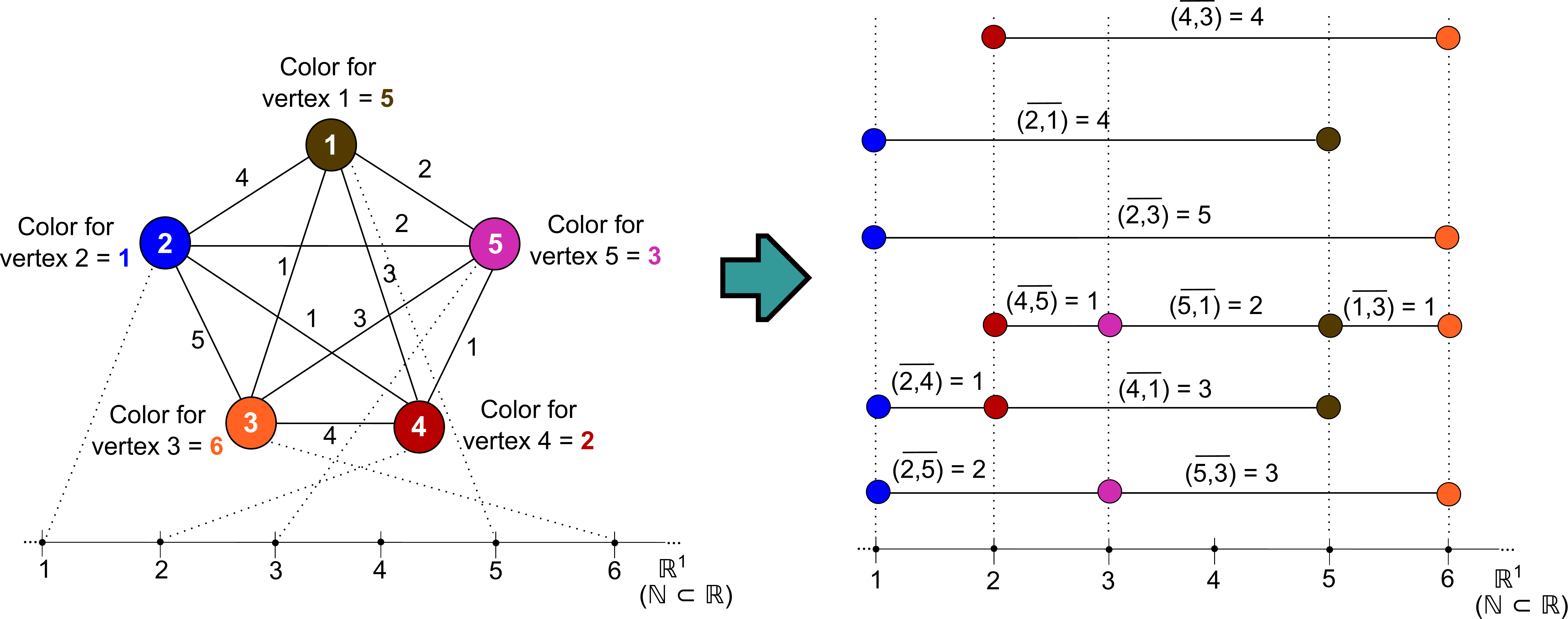}
\caption{Example from Figure \ref{fig:exNetDist} using 0-spheres (line segments).}
\label{fig:colorLines}
\vspace{0.9cm}
\end{figure}

\begin{figure}
\centering
\includegraphics[scale=0.21]{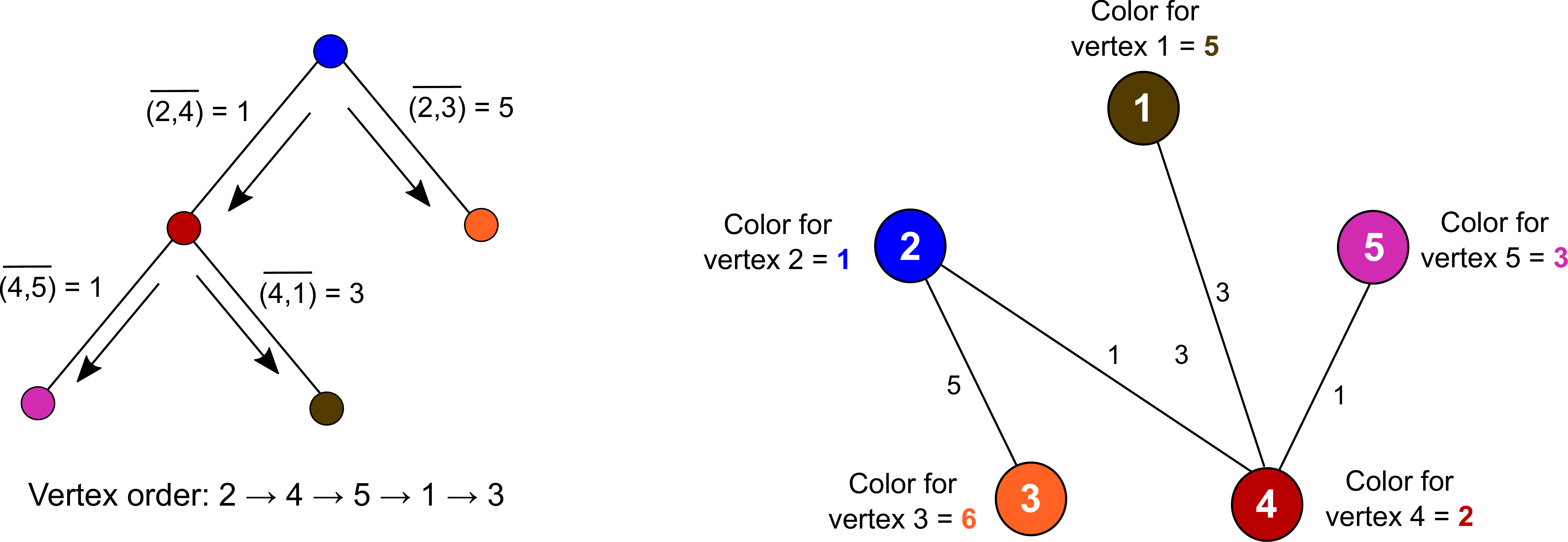}
\caption{Specific order of 0-spheres that leads to the optimal solution for Figure \ref{fig:exNetDist}.}
\label{fig:colorLines-2}
\vspace{1.0cm}
\end{figure}

In this work we focus on problems with exact distances between colors, and also in the analysis of different types of 
BPB algorithms and integer programming models.

Based on DMDGP, which is a decision problem involving equality distance constraints, the basic distance graph coloring 
model we consider also involves equality constraints between colors of two neighbor vertices $i$ and $j$. That is, the 
absolute difference between them must be exactly equal to an arbitrary weight imposed on the edge $(i, j)$, and
the solution candidate must satisfy all given constraints. We can formally define as follows.

\revisado{Given a graph $G=(V,E)$, we define $d_{i,j}$ as a positive integer weight associated to an edge $(i,j) \in E(G)$. 
In distance coloring, for each vertex $i$, a color must be determined for it (denoted by $x(i)$) such
that the constraints imposed on the edges between $i$ and its neighbors are satisfied.} A variation of the classic 
graph coloring problem consists in finding the minimum \textit{span} of $G$, that is, in determining that the maximum 
$x(i)$, or color used, be the minimum possible. Based on these preliminary definitions, we describe the following distance 
geometry vertex coloring problems.

\begin{defn}
\textbf{Coloring Distance Geometry Problem (CDGP):} Given a simple
weighted undirected graph $G = (V, E)$, where, for each $(i, j) \in E$, there is a
weight $d_{i,j} \in \mathbb{N}$, find an
embedding \revisado{$x: V \rightarrow \mathbb{N}$} (that is, an embedding of $G$ on the real number line, 
\revisado{but considering only the natural number points})
such that $|x(i) - x(j)| = d_{i,j}$ for each $(i, j) \in E$.
\end{defn}

CDGP involves equality constraints, and thus is named as Equal Coloring Distance Geometry Problem and labeled as \textbf{EQ-CDGP}. A solution for this problem consists of a tree, whose vertices are colored with colors that respect the \revisado
{equality} constraints involving the weighted edges (see Figure~\ref{fig:colorLines-2}). Since CDGP (or EQ-CDGP) is a \textit
{decision problem}, only a feasible solution is required. 
\revisado{This problem is NP-complete, as shown below.}

\revisado{\begin{thm}
\label{thm:EQ-CDGP-NP-comp}
EQ-CDGP is NP-complete.
\end{thm}
\begin{proof}
To prove that EQ-CDGP $\in$ NP-complete, we must show that EQ-CDGP $\in$ NP and EQ-CDGP $\in$ NP-hard.\\ \ \\
\noindent\textbf{1. EQ-CDGP \boldmath$\in$ NP.}\\
\noindent Given, for a graph $G = (V, E)$, an embedding \revisado{$x: V \rightarrow \mathbb{N}$}, its feasibility can be checked 
by taking each edge $(i, j) \in E$ and examining if its endpoints do not violate the corresponding distance constraint, that is, if $|x(i) - x(j)| = d_{i, j}$. If all distance constraints are valid, then $x$ is a
certificate for a positive answer to the EQ-CDGP instance, meaning that a certificate for a
YES answer can be verified
in $O(|E|)$ time, which is linear. Thus, EQ-CDGP $\in$ NP.\\ \ \\
\noindent\textbf{2. EQ-CDGP \boldmath$\in$ NP-hard.}\\
\noindent Since EQ-CDGP is equivalent to {\sc 1-Embeddability} with integer weights, 
which is NP-hard \citep{saxe:1979}, we can use
the same proof for the latter problem to show that EQ-CDGP is also NP-hard. The proof is made by
reducing the {\sc Partition} problem, which is known to be NP-complete \citep{garey:1979} to EQ-CDGP.\\
\noindent Consider a {\sc Partition} instance, consisting of a set $I$ of $r$ integers, that is, 
$M = \{m_1, m_2, \dots, m_r\}$. Let $G$ be a weighted graph $G = (V, E)$, where $G$ is a cycle
such that $|V| = |E| = r$ and,
for each edge $(i, j)$, its weight is a natural number denoted by $d_{i, j}$. This graph is constructed from $M$ by 
considering:
\begin{itemize}
\item $V = \{i_0, i_1, \dots, i_{r-1}\}.$
\item $E = \{(i_b, i_{b+1 \text{ mod } r})\ |\ 0 \le b \le r\}.$
\item $d_{i_b, i_{b+1 \text{ mod } r}} = m_b\ (\forall 0 \le b \le r).$
\end{itemize}
Now, let \revisado{$x: V \rightarrow \mathbb{N}$} be an embedding of the vertices on the number line. If it is
a valid embedding, then we can define two sets:
\begin{itemize}
\item $S_1 = \{m_b\ |\ x(i_b) < x(i_{b+1 \text{ mod } r})\}$.
\item $S_2 = \{m_b\ |\ x(i_b) > x(i_{b+1 \text{ mod } r})\}$.
\end{itemize}
We have that $S_1$ and $S_2$ are disjoint subsets of $M$ (that is, they form a partition of $M$) where
the sum of all $S_1$ elements is equal to the sum of all $S_2$ elements, that is,
if the cyclic graph constructed from $G$ admits an embedding on the line (which means that its solution
to EQ-CDGP is YES), then $M$ has a YES solution for {\sc Partition} and vice-versa. This reduction
can be made in $O(r)$ time, thus, EQ-CDGP $\in$ NP-hard.
\end{proof}}

\revisado{To illustrate the reasoning from Theorem \ref{thm:EQ-CDGP-NP-comp}, let $M$ be an instance of {\sc Partition}
such that $M =\{1, 4, 5, 6, 7, 9\}$. Figure \ref{fig:partition} shows its corresponding EQ-CDGP solution.}

\begin{figure}
\centering
\subfigure[][Instance with a YES solution.]
{\includegraphics[scale=0.36]{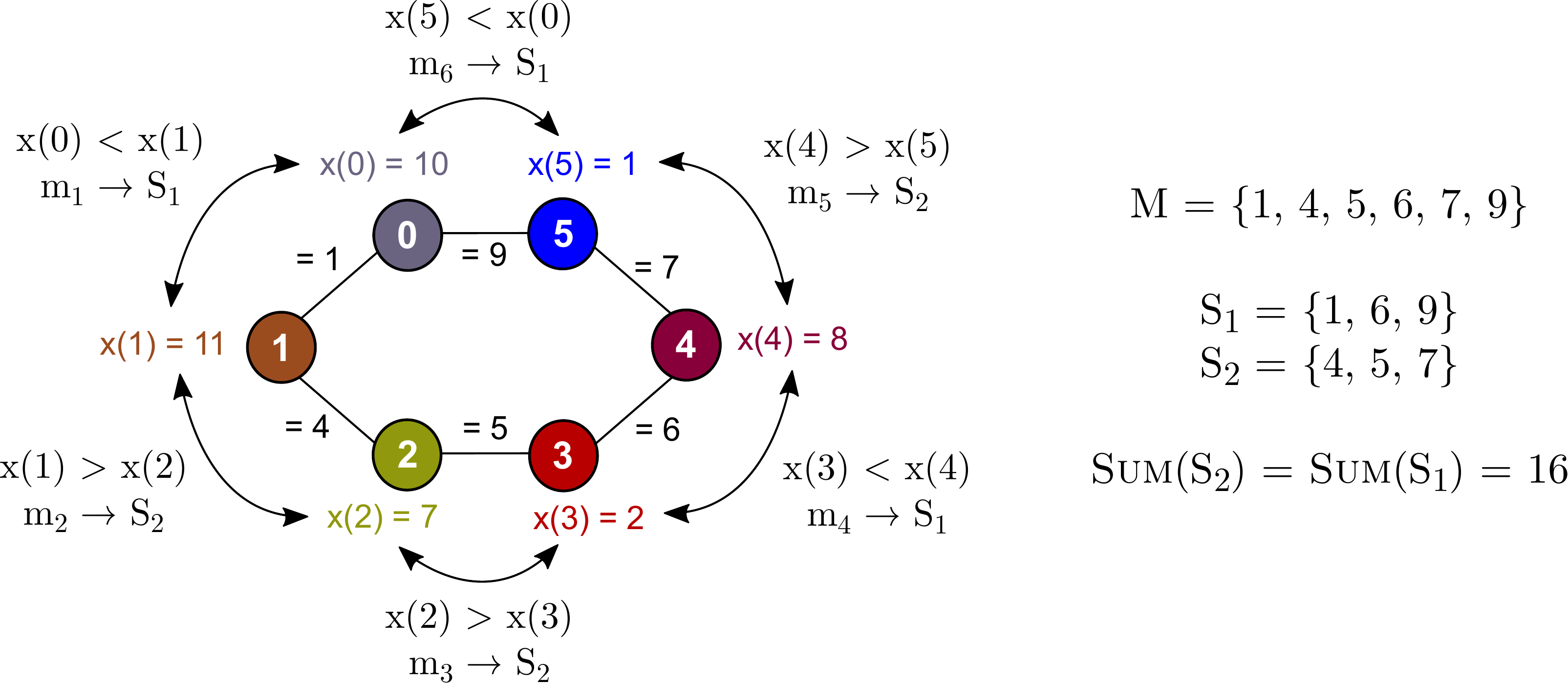}}
\ \\
\subfigure[][Instance with a NO solution.]
{\includegraphics[scale=0.36]{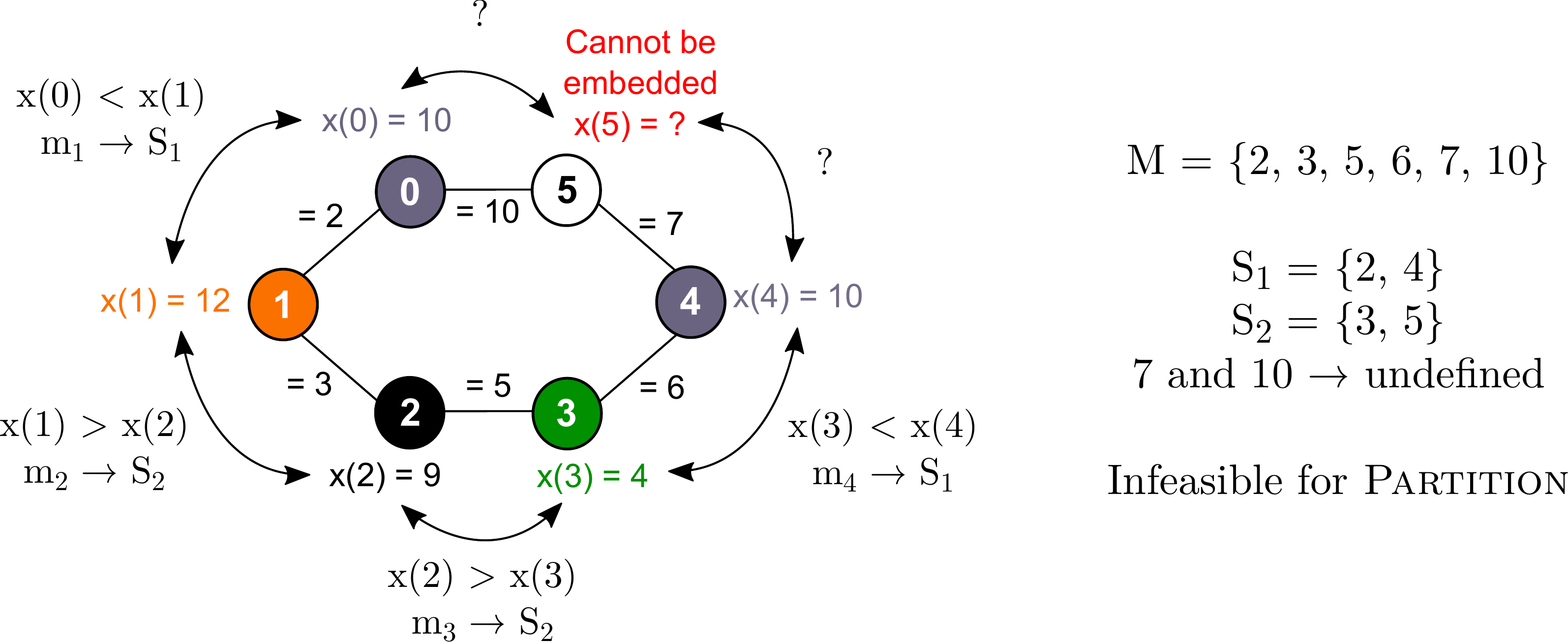}}
\caption{{\sc Partition} instances and corresponding transformations to EQ-CDGP.}
\label{fig:partition}
\end{figure}

Since most graph coloring problems in the literature and in real world applications are optimization problems,
we define an optimization version of this basic distance geometry graph coloring problem, as shown below.

\begin{defn}
\textbf{Minimum Equal Coloring Distance Geometry Problem (MinEQ-CDGP):} Given a simple
weighted undirected graph $G = (V, E)$, where, for each $(i, j) \in E$, there is a
weight $d_{i,j} \in \mathbb{N}$, find an
embedding \revisado{$x: V \rightarrow \mathbb{N}$}
such that $|x(i) - x(j)| = d_{i,j}$ for each $(i, j) \in E$ whose span $S$, defined as
$S = \max_{i \in V} x(i)$, that is, the maximum
used color, is the minimum possible.
\end{defn}

%\riscado{Thus, in this case, a solution is the best possible feasible solution, that is, a tree of the graph $G$ that satisfies the constraints with the minimum span.}

Figure \ref{fig:mineq-CDGP-0sp} shows an example of this model and its corresponding 0-sphere
visualization.

\begin{figure} [H]
\centering
\includegraphics[scale=0.42]{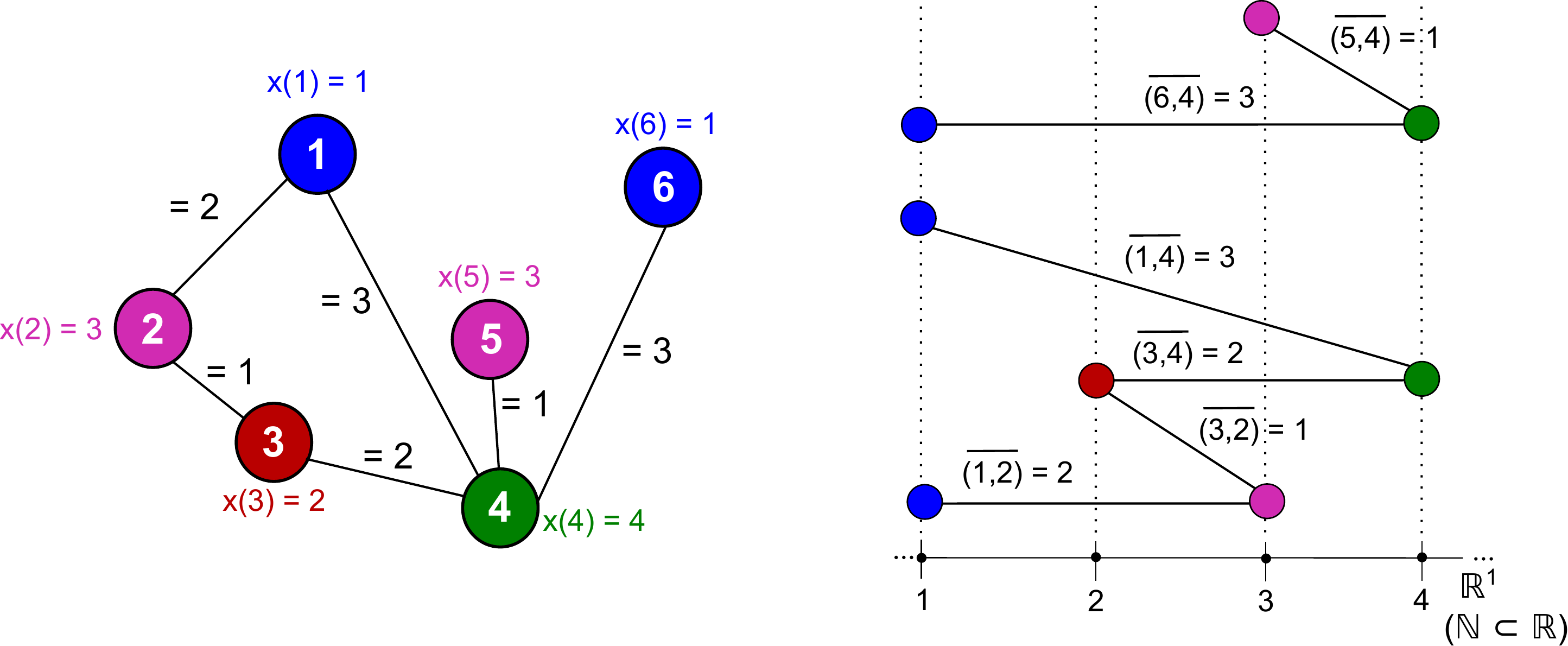}
\caption{MinEQ-CDGP instance with solution and its 0-sphere representation.}
\label{fig:mineq-CDGP-0sp}
\end{figure}

\revisado{On} the other hand, instead of equalities, we can consider inequalities, such that
the weight $d_{i,j}$ on an edge $(i, j)$ is actually a lower bound for the distance to be respected between the color points $x(i)$ and $x(j)$, that is, $|x(i) - x(j)| \geq d_{ij}$. Thus, we can modify MinEQ-CDGP to deal with this scenario, which becomes the following model.

\begin{defn}
\textbf{Minimum Greater than \revisado{or} Equal Coloring Distance Geometry Problem\linebreak
(MinGEQ-CDGP):} Given a simple
weighted undirected graph $G = (V, E)$, where, for each $(i, j) \in E$, there is a
weight $d_{i,j} \in \mathbb{N}$, find an
embedding \revisado{$x: V \rightarrow \mathbb{N}$} such that $|x(i) - x(j)| \ge d_{i,j}$ for each $(i, j) \in E$
whose span ($\max\limits_{i \in V} x(i)$) is the minimum possible.
\end{defn}

MinGEQ-CDGP is equivalent to the bandwidth coloring problem (BCP) \citep{malaguti:2010}, which
itself is equivalent to the minimum span frequency assignment problem (MS-FAP) \citep{koster:1999,audhya:2011}.

Figure \ref{fig:mineq-CDGP-0sp}
In Figure \ref{fig:mingeq-CDGP-0sp}, this model, along with its 0-sphere representation, is exemplified.

\begin{figure} [H]
\centering
\includegraphics[scale=0.42]{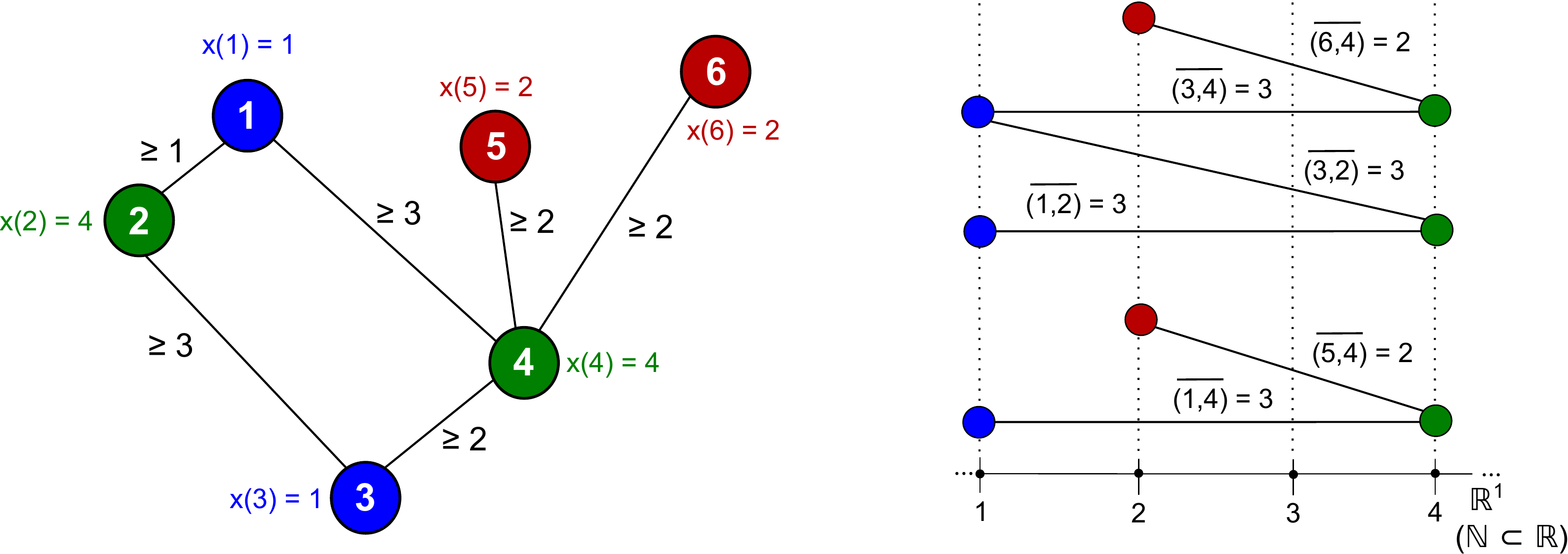}
\caption{MinGEQ-CDGP instance with solution and its 0-sphere representation.}
\label{fig:mingeq-CDGP-0sp}
\end{figure}

%Figure \ref{fig:cdgp-arbitrary} shows examples of the proposed distance coloring models.
%
%\begin{figure} [H]
%\centering
%\subfigure[][MinEQ-CDGP.]{\includegraphics[scale=0.42]{images/ex-ECDGP-arbDist.pdf}}
%\hspace*{1.6cm}
%\subfigure[][MinGEQ-CDGP.]{\includegraphics[scale=0.42]{images/ex-MCDGP-arbDist.pdf}}\\
%\caption{Examples of instances for distance coloring models and feasible solutions for them.}
%\label{fig:cdgp-arbitrary}
%\end{figure}

\subsection{Special cases}

For the models previously stated, we can identify some specific scenarios for which additional properties can be identified. The first special case is for EQ-CDGP, \revisado{the decision distance coloring problem}, where all distances are the same, \revisado{that is, the input is} a graph with uniform edge weights, as stated below.
\revisado{
\begin{defn}
\textbf{Coloring Distance Geometry Problem with \revisado{Uniform Distances} (EQ-CDGP\revisado{-Unif}):} Given a simple
weighted undirected graph $G = (V, E)$, and a nonnegative integer $\varphi$,
find an embedding \revisado{$x: V \rightarrow \mathbb{N}$} such that $|x(i) - x(j)| = \varphi$ for each $(i, j) \in E$.
\end{defn}}

\revisado{For the optimization version, we can also define this special case, as shown below.}

\begin{defn}
\textbf{Minimum Equal Coloring Distance Geometry Problem with \revisado{Uniform Distances} (MinEQ-CDGP\revisado{-Unif}):} Given a simple
weighted undirected graph $G = (V, E)$, and a nonnegative integer $\varphi$,
find an embedding \revisado{$x: V \rightarrow \mathbb{N}$} such that $|x(i) - x(j)| = \varphi$ for each $(i, j) \in E$
whose span ($\max_{i \in V} x(i)$) is the minimum possible.
\end{defn}

In this model, an input graph can be defined by its sets of vertices and edges and the $\varphi$ value, instead
of a set of weights for each edge. A similar special case exists for MinGEQ-CDGP, as stated in the following definition.

\begin{defn}
\textbf{Minimum Greater than \revisado{or} Equal Coloring Distance Geometry Problem with \revisado{Uniform Distances} (MinGEQ-CDGP\revisado{-Unif}):}
Given a simple
weighted undirected graph $G = (V, E)$, and a nonnegative integer $\varphi$,
find an embedding \revisado{$x: V \rightarrow \mathbb{N}$} such that $|x(i) - x(j)| \ge \varphi$ for each $(i, j) \in E$
whose span ($\max_{i \in V} x(i)$) is the minimum possible.
\end{defn}

When $\varphi=1$, MinGEQ-CDGP\revisado{-Unif} is equivalent to the classic graph coloring problem (Figure~\ref{fig:cdgp-const}).

A summary of all distance coloring models, including special cases, is given in Table \ref{tbl:dist-col}. \\

\begin{table} [h]
	\centering
	\caption{Summary of distance coloring models.}
	\label{tbl:dist-col}
	\scalebox{0.95}{
		\begin{tabular}{ c>{\centering}m{5cm}>{\centering}m{5cm}}
			\hline
			\textbf{Problem} & \textbf{Constraint type} & \textbf{Distance type} \tabularnewline
			\hline
			EQ-CDGP and MinEQ-CDGP  &
			$\forall (i, j) \in E,\ |x(i)-x(j)| = d_{i,j}$ & $\forall (i, j) \in E,\ d_{i,j} \in \mathbb{N}$ \tabularnewline
		
			MinGEQ-CDGP &
			$\forall (i, j) \in E,\ |x(i)-x(j)| \ge d_{i,j}$ & $\forall (i, j) \in E,\ d_{i,j} \in \mathbb{N}$ \tabularnewline
			
			EQ-CDGP\revisado{-Unif} and MinEQ-CDGP\revisado{-Unif} &
			$\forall (i, j) \in E,\ |x(i)-x(j)| = d_{i,j}$ &
			\parbox{5cm}{\vspace*{1mm}\centering $\forall (i, j) \in E,\ d_{i,j} = \varphi$\\$(\varphi \in \mathbb{N})$\vspace*{1mm}} \tabularnewline
			
			MinGEQ-CDGP\revisado{-Unif} &
			$\forall (i, j) \in E,\ |x(i)-x(j)| \ge d_{i,j}$ &
			\parbox{5cm}{\vspace*{1mm}\centering $\forall (i, j) \in E,\ d_{i,j} = \varphi$\\$(\varphi \in \mathbb{N})$\vspace*{1mm}} \tabularnewline
			\hline
%			\revisado{Graceful-GL} &
%			\revisado{$\forall (i, j) \in E,\ |x(i)-x(j)| = d_{i,j}$} & 
%			\parbox{5cm}{\vspace*{1mm}\centering \revisado{$\forall (i, j) \in E,\ d_{i,j} \in \mathbb{Z}^*$\\
%			$\forall (i, j), (u, v) \in E,\ d_{i,j} \neq d_{u, v}$}\vspace*{1mm}} \tabularnewline
%			\hline
		\end{tabular}}
	\end{table}

\begin{figure} [H]
	\centering
	\subfigure[][MinEQ-CDGP\revisado{-Unif}.]{\includegraphics[scale=0.42]{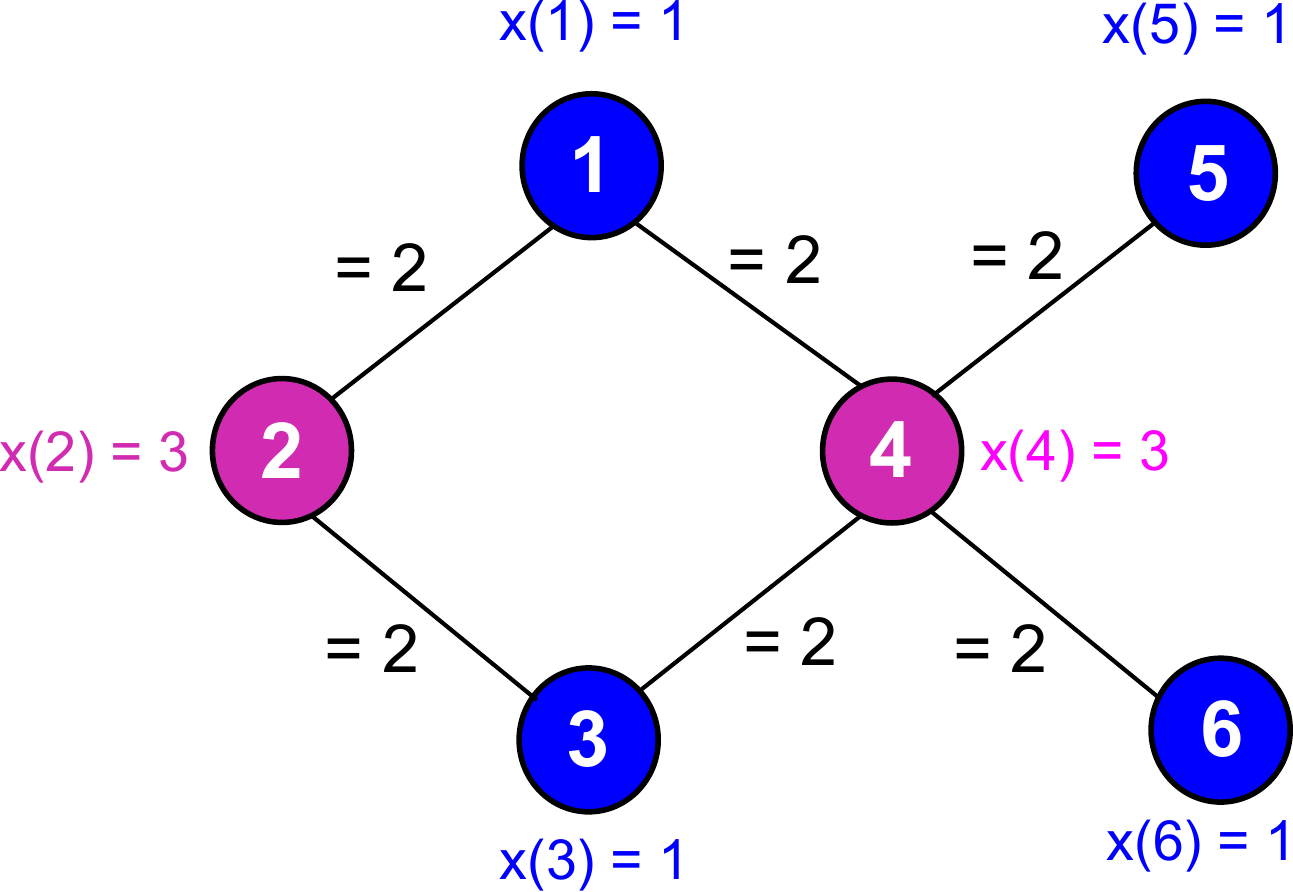}}
	\hspace*{1.6cm}
	\subfigure[][MinGEQ-CDGP\revisado{-Unif}.]{\includegraphics[scale=0.42]{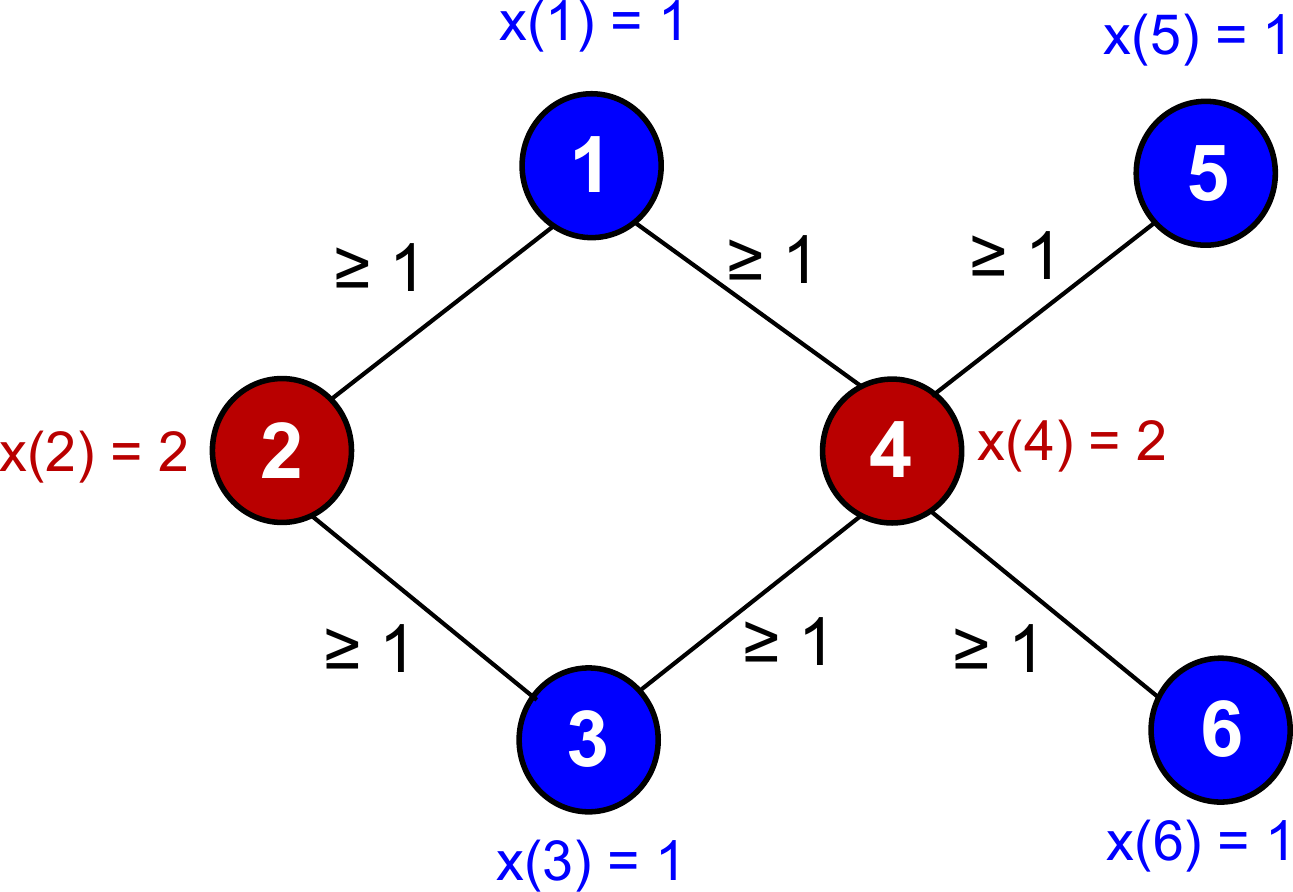}}
	\caption{Examples of instances for the special cases of distance coloring models with
		constant edge weights and feasible solutions for them.}
	\label{fig:cdgp-const}
\end{figure}

%\begin{figure}[H]
%\centering
%\includegraphics[scale=0.3]{images/graceful-to-ecdgp-2.pdf}
%\caption{Graceful labeling of a graph seen as a feasible solution of a
%MinEQ-CDGP instance.}
%\label{fig:graceful-to-ecdgp}
%\end{figure}

\section{Feasibility conditions of distance graph coloring problems}
\label{sec:distColProp}

In this section, we discuss feasibility conditions related to our proposed EQ-CDGP problems. Clearly, the problems involving inequality constraints \revisado{are} always feasible. This is the case for the MinGEQ-CDGP and MinGEQ-CDGP problems (and the special cases with uniform distances, MinGEQ-CDGP\revisado{-Unif} and MinGEQ-CDGP\revisado{-Unif}). However, this is not so for versions that involve only equality constraints, EQ-CDGP and its special case with uniform distances, the EQ-CDGP\revisado{-Unif} problem.

\subsection{Feasibility conditions for EQ-CDGP\revisado{-Unif}}

Graphs that admit a solution for the EQ-CDGP\revisado{-Unif} problem are characterized by the following 
theorem.

\begin{thm}
\label{thm:EQ-CDGP-Unif}
A graph $G$ has solution YES for EQ-CDGP\revisado{-Unif} problem if and only if $G$ is bipartite.
\end{thm}

\begin{proof} Let $G$ be a graph, input to a EQ-CDGP\revisado{-Unif} problem, where for each edge $v_iv_j$ of $G$, the distance required is $d_{ij}=\varphi$, $\varphi \in \mathbb{N}$, constant. Suppose $G$ \revisado{has a YES solution for the problem such that \revisado{$x: V \rightarrow \mathbb{N}$} is a certificate for that solution}. Let $x(i)$ be the color assigned to $v_i \in V$. Choose an arbitrary path $v_1,v_2,...,v_k$ of $G$, not necessarily simple. Then $|x(i)-x(j)|=\varphi$, for $|i-j|=1$. \problema{The latter implies $x(i)=x(i+2)$, $i=1,2,...,k-2$. Consequently, if the path contains the same vector $v_i$ twice, their corresponding indices are the same.} That is, \problema{all edges of $G$ are necessarily even}, and $G$ is bipartite.\\

\noindent Conversely, if $G$ is bipartite, its vertices admit a proper coloring with two distinct colors. Assign the value $x(i)$ to the vertices of the first color, and the value $\varphi + 1$ to the second one. Then $|x(i)-x(j)|=\varphi$, for each edge $v_iv_j$ of $G$, and EQ-CDGP\revisado{-Unif} \revisado{has a YES solution}.
%\end{proof}
%
%\begin{proof}
%\revisado{\textbf{**PROVA ALTERNATIVA**}
\noindent \revisado{As an alternative way of proving that if a graph is bipartite then it has a YES solution for EQ-CDGP, observe that, since the 
input graph is bipartite, it is also 2-colorable (considering the classic graph coloring problem), 
that is, the entire graph can be colored using only two
different colors, which can be determined by considering a single edge from the graph.
All edges $(i, j)$ have the same distance constraint, that is, $|x(i) - x(j)| = \varphi$, so
the two colors that will be used are \{1, 1+$\varphi$\}, which form the solution for the EQ-CDGP-Unif instance.\\}

\noindent \revisado{In order to prove the converse statement, that is, if a graph has a YES solution for EQ-CDGP, it is bipartite, we will use a proof by contrapositive, which states that if a graph is not bipartite, then it has a NO solution 
for EQ-CDGP. This will be done by mathematical induction on odd cycles, since a graph
is not bipartite if, and only if, it contains an odd cycle.
Let $|V| = 2z+1$. The proof will be by induction on $z$ (the number of vertices).\vspace*{2mm}}

\noindent\revisado{\textbf{Base case:} $z=1$. We have the cycle $C_3$, with
three vertices ($V = \{1, 2, 3\}$) and three edges
($\{(1,2), (1,3), (2,3)\}$), with $|x(i)-x(j)| = \varphi$ for all of them. Without loss of generality,
let $x(1) = 1$ and $x(2) = 1+\varphi$. Then we have that:
\begin{itemize}
\item Since $(1,3) \in E$ and $x(1) = 1$, then $|x(3)-1| = \varphi$. All colors must be
positive integers, so $x(3) = 1+\varphi$.
\item Since $(2,3) \in E$ and $x(2) = 1+\varphi$, then $|x(3)-(1+\varphi)| = \varphi\ \ \Rightarrow\ \
|x(3)-1-\varphi| = \varphi$. By this inequation, $x(3) = 1$ or $x(3) = 1+ 2\varphi$.
\end{itemize}
From this result, we have that $x(3) = 1+\varphi$ and ($x(3) = 1$ or $x(3) = 1+ 2\varphi$) at the same time, which is impossible.
Then $C_3$ has a NO solution for EQ-CDGP, as seen in Figure \ref{fig:c3}. \vspace*{2mm}}

\noindent\revisado{\textbf{Induction hypothesis:} The cycle $C_{2z+1}$ has a NO solution for EQ-CDGP.\vspace*{2mm}}

\noindent\revisado{\textbf{Inductive step:} By the inductive hypothesis, the cycle $C_{2z+1}$ is infeasible
for EQ-CDGP. If we consider the cycle $C_{2(z+1)+1} = C_{2z+3}$, we have that the size of the cycle increases
by two vertices, but it will still be an odd cycle. If we add only one vertex, that is, we consider the
cycle $C_{2z+1+1} = C_{2z+2}$, we will have an even cycle. Since all even cycles are bipartite, they
are feasible in EQ-CDGP according to Theorem \ref{thm:EQ-CDGP-Unif}. Now, consider that another vertex is added to
$C_{2z+2}$, becoming $C_{2z+3}$. Without loss of generality, consider that the new vertex
$2z+3$ is adjacent to vertices $2z+2$ and $1$, that is, we have $\{(2z+2, 2z+3), (2z+3, 1)\} \subseteq E$,
and $x(2z+2) = 1+\varphi$ and $x(1) = 1$ (these colors can be seen as having been assigned when we added
only one vertex, generating an even cycle). Then we have that:
\begin{itemize}
\item Since $(2z+2, 2z+3) \in E$ and $x(2z+2) = 1+\varphi$, then $|x(2z+3)-(1+\varphi)| = \varphi\ \ \Rightarrow\ \ |x(2z+3)-1-\varphi| = \varphi$. By this inequation, $x(2z+3) = 1$ or $x(2z+3) = 1+ 2\varphi$.
\item Since $(2z+3, 1) \in E$ and $x(1) = 1$, then $|x(2z+3)-1| = \varphi$. All colors must be
positive integers, so $x(2z+3) = 1+\varphi$.
\end{itemize}
From this result, we have that $x(2z+3) = 1+\varphi$ and ($x(2z+3) = 1$ or $x(2z+3) = 1+ 2\varphi$) at the same time, which is impossible. Therefore $C_{2z+3}$ has a NO solution for EQ-CDGP, as can be seen in Figure \ref{fig:c-odd}.}
\end{proof}

\begin{figure}
\centering
\includegraphics[scale=0.5]{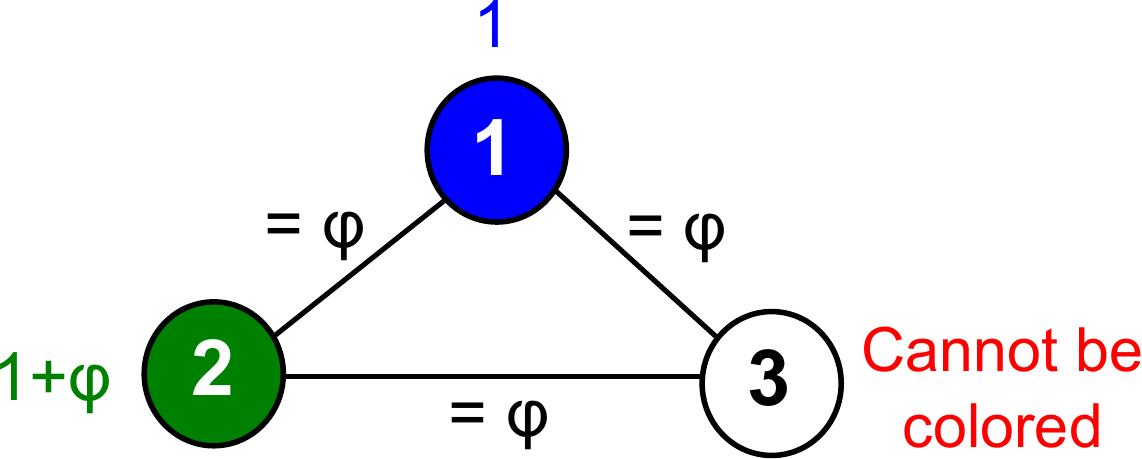}
\caption{$C_3$ graph that has a NO solution for EQ-CDGP-Const when all distances are the same.}
\label{fig:c3}
\end{figure}

\begin{figure}
\centering
\includegraphics[scale=0.5]{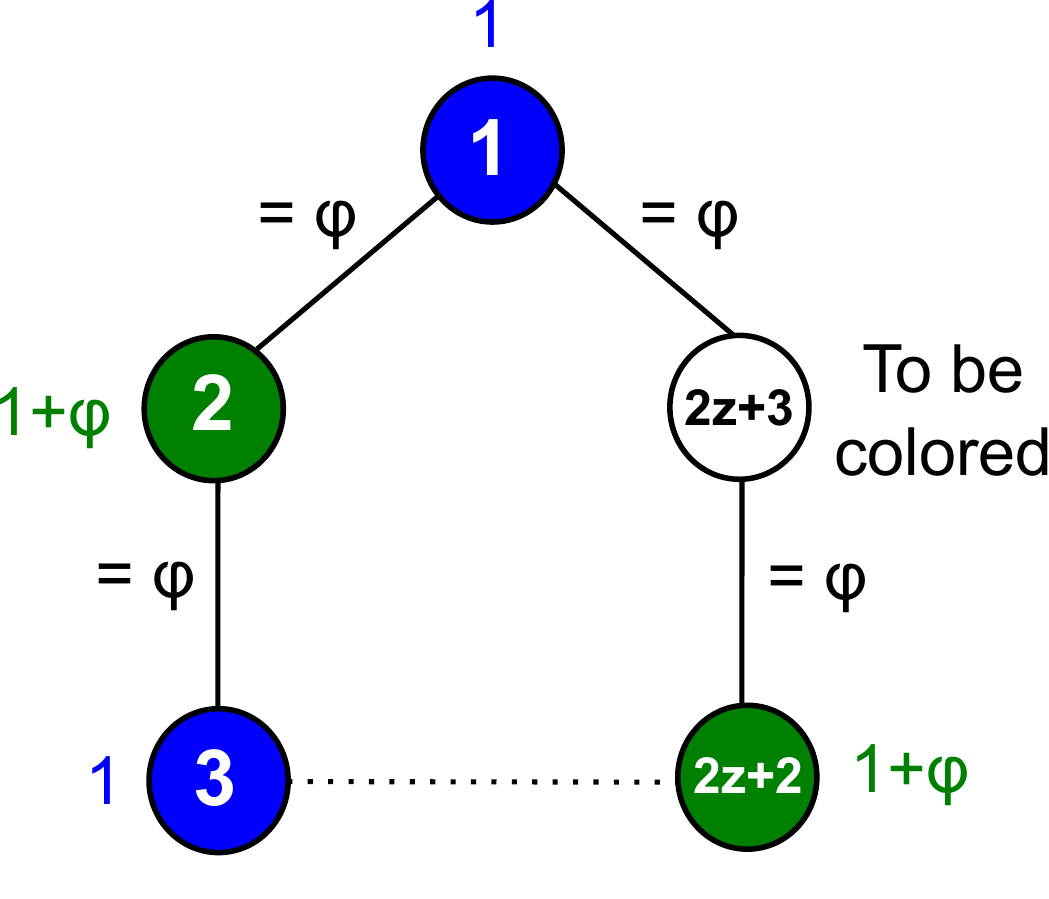}
\caption{Odd cycle $C_{2z+3}$ that has a NO solution for EQ-CDGP-Const when all distances are the same.}
\label{fig:c-odd}
\end{figure}

\revisado{As a complementary result, graphs which have odd-length cycles as induced subgraphs will always have
a NO solution for EQ-CDGP-Unif, because a graph is bipartite if, and only if, it contains no odd-length cycles.}
Since the recognition of bipartite graphs can be done in linear time \revisado{using a graph search algorithm
such as depth-first search (DFS)}, the EQ-CDGP\revisado{-Unif} problem can be solved in 
\revisado{linear} time.

\subsection{Feasibility conditions for EQ-CDGP}

Clearly, Theorem~\ref{thm:EQ-CDGP-Unif} does not apply when the distances are arbitrarily defined. 
\revisado{Depending on the edge weights, bipartite graphs may have NO solutions for EQ-CDGP, and graphs
which include odd-length cycles may have YES solutions. Figure \ref{fig:cycles} shows examples
of instances considering each case.} However, this decision problem can be easily solved for trees, as shown below.

\begin{figure}
	\centering
	\subfigure[][Bipartite graphs.]{\includegraphics[scale=0.4]{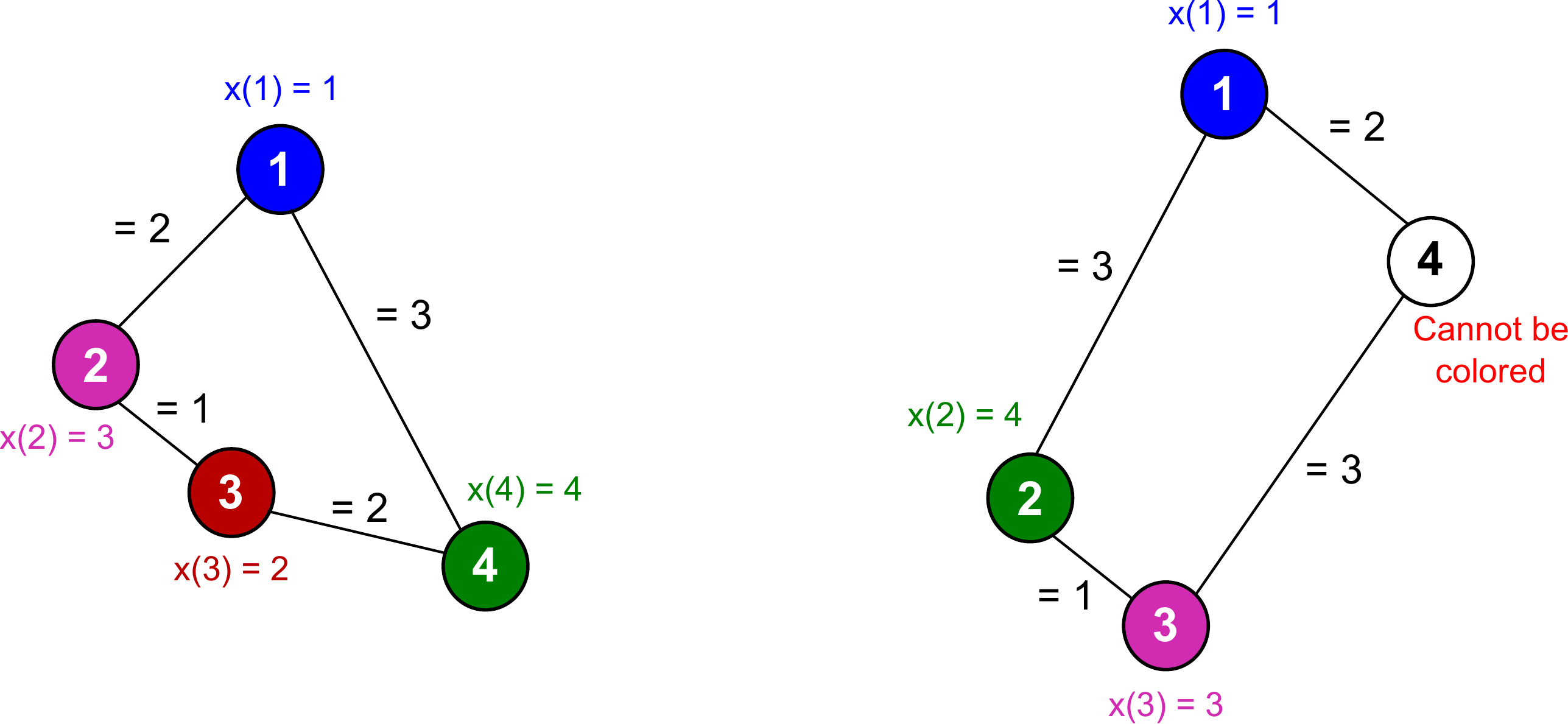}}
	\subfigure[][Graphs with odd-length cycles.]{\includegraphics[scale=0.4]{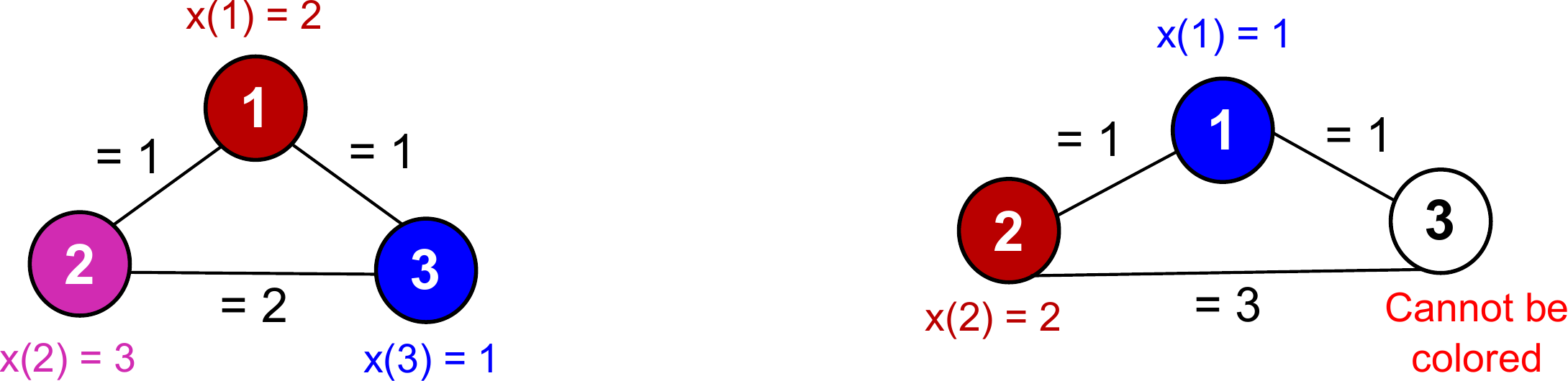}}
	\caption{Examples of instances for the special cases of distance coloring models with
		constant edge weights and feasible solutions for them.}
	\label{fig:cycles}
\end{figure}

\begin{thm}
	\label{thm:EQ-CDGParbdist}
Let $G = (V, E, d)$, be a tree, where $\forall (i, j) \in E\ \ d_{i, j}$ is an arbitrary positive integer. Then $G$ \revisado{always has a YES solution for EQ-CDGP}.
\end{thm}
%\pagebreak

\begin{proof}
We describe a simple algorithm for assigning colors that satisfy the EQ-CDGP problem.

\noindent Initially unmark all vertices. Choose an arbitrary vertex $v_i$, assign any positive integer value $x(v_i)$ to $v_i$, and mark $v_i$. Iteratively, choose an unmarked vertex $v_j$, adjacent to some marked vertex $v_k$. Assign the value $x(v_j)=d_{jk}+x(v_j)$ and mark $v_j$. Repeat the iteration until all vertices become marked.
\end{proof}

The algorithm described in Theorem~\ref{thm:EQ-CDGParbdist} has linear time complexity. It is important to note that
\revisado{when this procedure is applied to a MinEQ-CDGP instance, that is, the optimization problem
with equality constraints, it only guarantees that a feasible solution is found for a tree}, not the optimal one.

%\begin{figure}[t]
%\centering
%\includegraphics[scale=0.47]{images/arvore-fact}
%\caption{Path decomposition in a tree.}
%\label{fig:arvore-fact}
%\end{figure}

\section{Algorithmic techniques and methods to solve EQ-CDGP models}
\label{sec:algs}

In this section, we show some algorithmic strategies to solve our distance geometry graph coloring models, and discuss some algorithmic strategies considering the EQ-CDGP models proposed in the previous section.

\subsection{Branch-prune-and-bound methods}
\label{sec:bpb}

For solving the three distance geometry graph coloring models shown in Section~\ref{sec:probStat}, we developed three
algorithms that combine concepts from constraint propagation and optimization techniques.

A branch-and-prune (BP) algorithm was proposed by \cite{lavor:2012:1} for the Discretizable Molecular Distance Geometry Problem (DMDGP), based on a previous version for the MDGP by \cite{liberti:2008}. The algorithm proceeds by \revisado{enumerating} \problema{possible positions for the vertices that must be located in three-dimensional space ($\mathbb{R}^3$), by manipulating the set of available distances}. The position for a vertex $i$, where $i \in [4, n]$ and $n$ is the number of vertices that must be placed in $\mathbb{R}^3$, is determined with respect to the last three vertices that have already been positioned, following the ordering and sphere intersection cited in Section~\ref{sec:probStat}.
However, a distance between the currently positioned vertex and a previous one that was placed before the last three can be violated, which requires feasibility tests to guarantee that the solution being built is valid. The authors applied
the Direct Distance Feasibility (DDF) pruning test, where
$\forall (i, j) \in E\ \ |||x(i) - x(j)|| - d_{i, j}| < \epsilon$, and where $\epsilon$ is a given tolerance.

In this work, we \revisado{adapted} these concepts \revisado{to study and solve our} proposed distance geometry coloring models. One of the first reflections
that can be made is that for the distance geometry coloring models, there \revisado{are} no initial assumptions to be \revisado{respected}, and thus, there is no explicit vertex ordering to be considered, so we build the ordering by an implicit enumerating process. We mix concepts from branch-and-prune for DMDGP and branch-and-bound procedures
to obtain partial solutions (sequences of vertices that have already been colored) that cannot improve on the current best solution. 

Our \textit{branch-prune-and-bound} (BPB) method works as follows. 
First, a vertex $i$ \revisado{that has not been colored yet} is selected as a starting point.
This vertex receives the color $1$, which is the lowest available (since all colors are positive integers).
Then a neighbor $j$ of $i$ that has not been colored yet is selected. A color selection algorithm
is used for setting a color to $j$ and the process is repeated recursively for neighbors of $j$ that
have not been colored yet. When an uncolored neighbor of the current vertex cannot be found, a uncolored
vertex of the graph is used. Pseudocode for this general procedure is given in Algorithm~\ref{alg:bpb1}. 

\begin{algorithm} [t]
\caption{\ Branch-prune-and-bound general algorithm.}
\label{alg:bpb1}
\begin{spacing}{0.9}
\small
\begin{algorithmic}[1]
\Require{graph $G$ (with set $V$ of vertices and set $E$ of edges), function $d: E \rightarrow \mathbb{N}$ of 
distances for each edge, previous vertex $i$, current vertex $j$ to be colored, current partial coloring $x$, best complete coloring found $x_{best}$, upper bound $ub$, array $pred$ of predecessors from each vertex
(initially all set to -1) and enumeration tree depth $dpt$.}
\Function{Branch-Prune-And-Bound}{$G = (V, E), d, i, j, x, x_{best}, ub, pred, dpt$}
	\For{\textbf{each} neighbor $k$ of $j$}
		\If{$predec[k] = -1$}
			\State $predec[k] \gets i$ \Comment{Set current vertex as predecessor of neighbors}
		\EndIf
	\EndFor
	\If{$i = -1$}
		\State $i \gets predec[j]$ \Comment{If this call did not come from a neighbor, use predecessor information}
	\EndIf
	\State $colorsAvail \gets$ \Call{SelectColors}{$G, d, i, j, x, ub$}
	\While{$colorsAvail \neq \emptyset$}
		\State $color \gets$ element of $colorsAvail$
		\State $colorsAvail \gets colorsAvail - \{color\}$
		\State $x(j) \gets color$
		\If{$\max\limits_{v \in V\ |\ v \text{ is colored}} x(v) \ge ub$}
			\State Remove color from $i$
			\State \textbf{continue} \Comment{Discard this possible partial solution by bounding}
		\EndIf
		\If{\Call{FeasibilityTest}{$G, d, f, x, i$} = \textbf{false}}
			\State Remove color from $i$
			\State \Return \Comment{Distance violation, discard partial solution by pruning}
		\EndIf
		\If{$dpt = |V|$} \Comment{If true, then all vertices are colored}
			\If{$\max\limits_{v \in V} x(v)$ $<$ $\max\limits_{v \in V} x_{best}(v)$}
				\State $x_{best} \gets x$
				\State \revisado{$ub \gets \max\limits_{v \in V} x(v)$}
			\EndIf
		\Else
			\State $hasNeighbor \gets$ \textbf{false}
			\For{\textbf{each} neighbor $k$ of $j$}
				\If{$k$ is not colored}
					\State $hasNeighbor \gets$ \textbf{true}
					\State \Call{Branch-Prune-And-Bound}{$G, d, f, j, k, x, x_{best}, ub, dpt+1$}
				\EndIf
			\EndFor
			\If{$hasNeighbor =$ \textbf{false}}
				\For{\textbf{each} vertex $k$ of $G$ such that $predec[k] \neq -1$} \Comment{Only from vertices with predecessors}
					\If{$k$ is not colored}
						\State \Call{Branch-Prune-And-Bound}{$G, d, f, -1, k, x, x_{best}, ub, dpt+1$}
					\EndIf
				\EndFor
			\EndIf
			\State Remove color from $i$
		\EndIf
	\EndWhile
	\State \Return $x_{best}$
\EndFunction
\end{algorithmic}
\end{spacing}
\vspace{0.2cm}
\end{algorithm}

We propose different strategies for selecting a color for a vertex and illustrate how the feasibility checking can be done in different levels of the procedure. Each of these cases are discussed below. \\

\vspace{0.3cm}

\noindent \textbf{Color selection for a vertex}\\

There are two possibilities for determining which colors a vertex can use\revisado{, determined by the call to \textsc{SelectColors}()), which returns a set of possible colors for a vertex}.

The first one\revisado{, denoted by BPB-Prev,} is based on the original BP algorithm by \cite{lavor:2012:1}. When a vertex $i$ has to be colored, the single previously colored vertex $j$ is taken into account. If $j$ is an invalid vertex,
which means that $i$ is not an uncolored neighbor of $j$, then the only color that $i$ can receive is
1. Otherwise, \revisado{the function returns a set of cardinality at most 2, whose elements are:}

\begin{enumerate}
\item $x(j) + d_{i, j}$.
\item $x(j) - d_{i, j}$ (returned only if $x(j) > d_{i, j}$).
%\vspace{0.5cm}
\end{enumerate}

%\vspace{1.0cm}

This means that this criterion uses only information from the previous vertex to determine colors, which makes
the BPB that uses it an inexact algorithm, something that the original BP for DMDGP also is \citep{lavor:2012:1}.
\revisado{However, to counter this in our BPB, when a vertex is colored, its neighbors are marked so that
they can use the current vertex as a predecessor in case the search restarts from one of such neighbors. Since
we assume the input graph is connected and the algorithm essentially walks through the graph, this information helps to find the true optimal solution. This procedure is done in $O(1)$ time, since only two
arithmetic operations are made to determine the colors.
An example of this color selection possibility is given in Figure \ref{fig:enum-Prev}.}

\revisado{When using this criterion, we apply the feasibility checking at each colored vertex. However, an alternative is to prune only infeasible solutions where all vertices have colors, that is, we apply the feasibility test only at the last level of the enumeration tree. An example of this alternative is shown in Figure \ref{fig:enum-Prev-FeasCheckFull}, where it is possible to see that this strategy makes the tree grow very large.}

The second selection criterion is undertaken using information from all colored neighbors to determine the color for the current
vertex $i$. This is done by solving a system of absolute value inequalities (or equalities, in
the case of MinEQ). Those inequalities arise from
the distance constraints for the edges. Let $i$ be the vertex that must be colored. The
color $x(i)$ must be the solution of a system of absolute value (in)equalities where there is one for
each colored neighbor $j$ and each one is as follows:

\begin{center}
$|x(j) - x(i)|\ \ \text{OP}\ \ d_{i, j}$
\end{center}

\noindent Where \textit{OP} is either ``='' (\revisado{for MinEQ-CDGP type problems})
or ``$\ge$'' (\revisado{for MinGEQ-CDGP type problems}). 
%according to the problem.
The color that will be assigned 
to $j$ is the smallest value that satisfies all (in)equalities.
We note that this procedure always returns \revisado{a set of cardinality 1, that is,} only one color (since only the lowest index is returned) which
is also feasible for the partial solution and eventually leads to the optimal solution, although it requires
more work per vertex. \revisado{This selection strategy runs in $O(ub)$ time, where $ub$ is an upper bound for the span,
since, to solve the system, we have to mark each possible solution in the interval $[1, ub]$ and select
the smallest value. Figure \ref{fig:enum-Select} shows an example of an enumeration tree using this
color selection strategy.}

\begin{figure}
\centering
\includegraphics[scale=0.2]{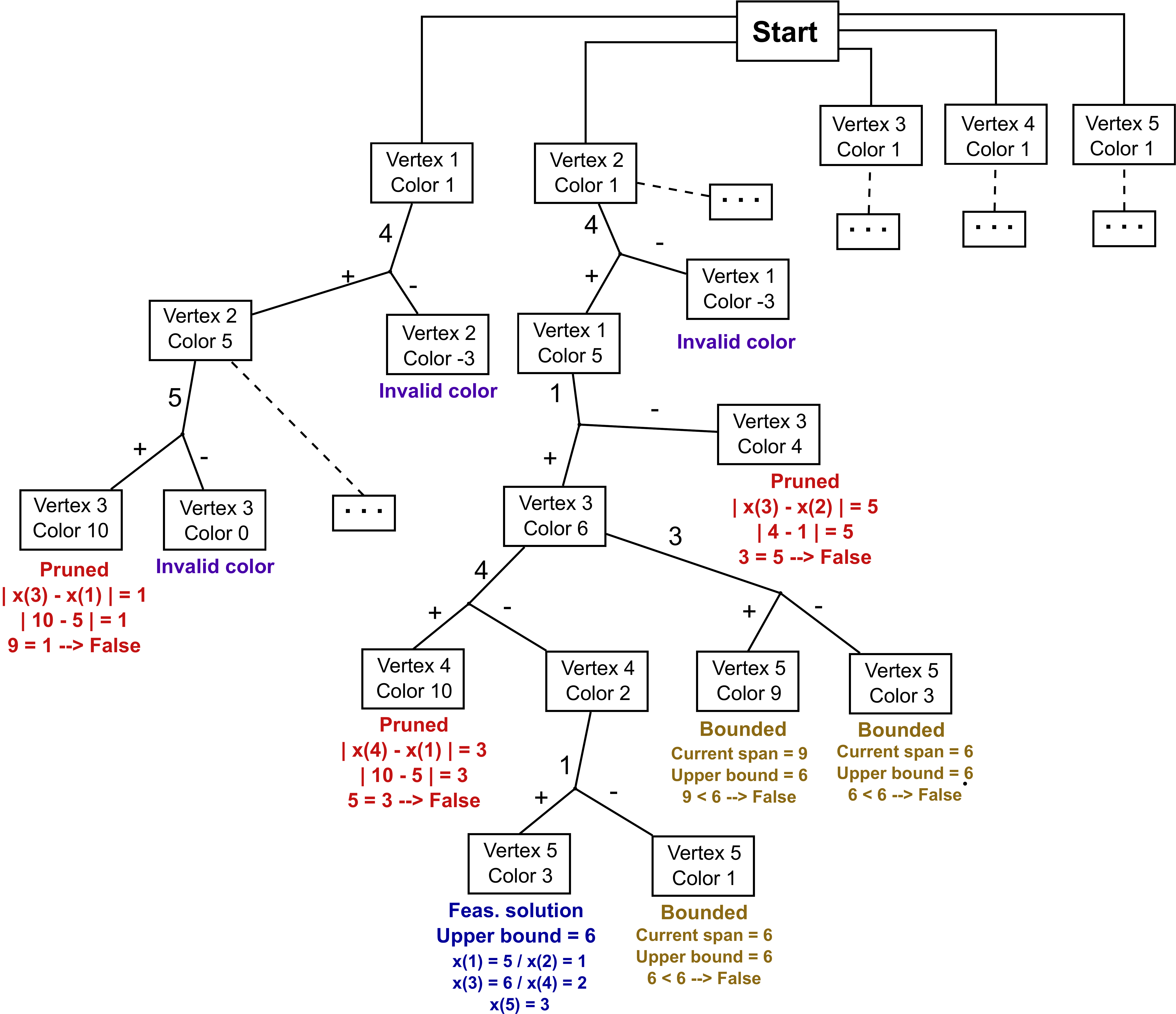}
\caption{Partial enumeration of solutions starting from vertex 2 for the MinEQ-CDGP instance defined by Figures \ref{fig:exNetDist} and \ref{fig:spheres} using BPB-Prev, with color selection based only on the previous vertex and feasibility checking at each partial solution.}
\label{fig:enum-Prev}
\end{figure}

\begin{figure}
\centering
\includegraphics[scale=0.187]{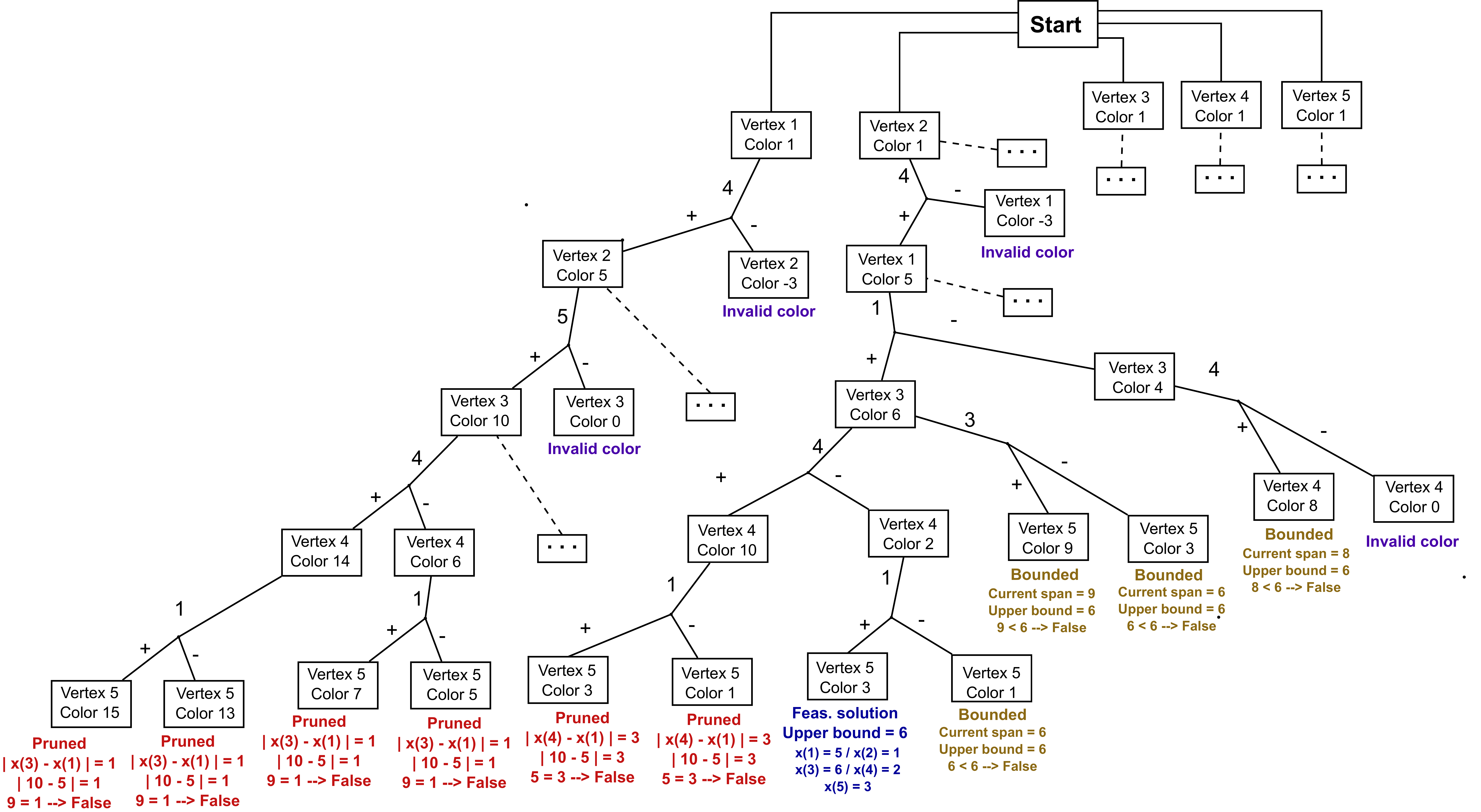}
\caption{Partial enumeration of solutions starting from vertex 2 for the MinEQ-CDGP instance defined by Figures \ref{fig:exNetDist} and \ref{fig:spheres} using BPB-Prev-FeasCheckFull with color selection based only on the previous vertex and feasibility checking only when all vertices are colored. 
The backtracking points are indicated when the solution is pruned.}
\label{fig:enum-Prev-FeasCheckFull}
\end{figure}

\begin{figure}
\centering
\includegraphics[scale=0.2]{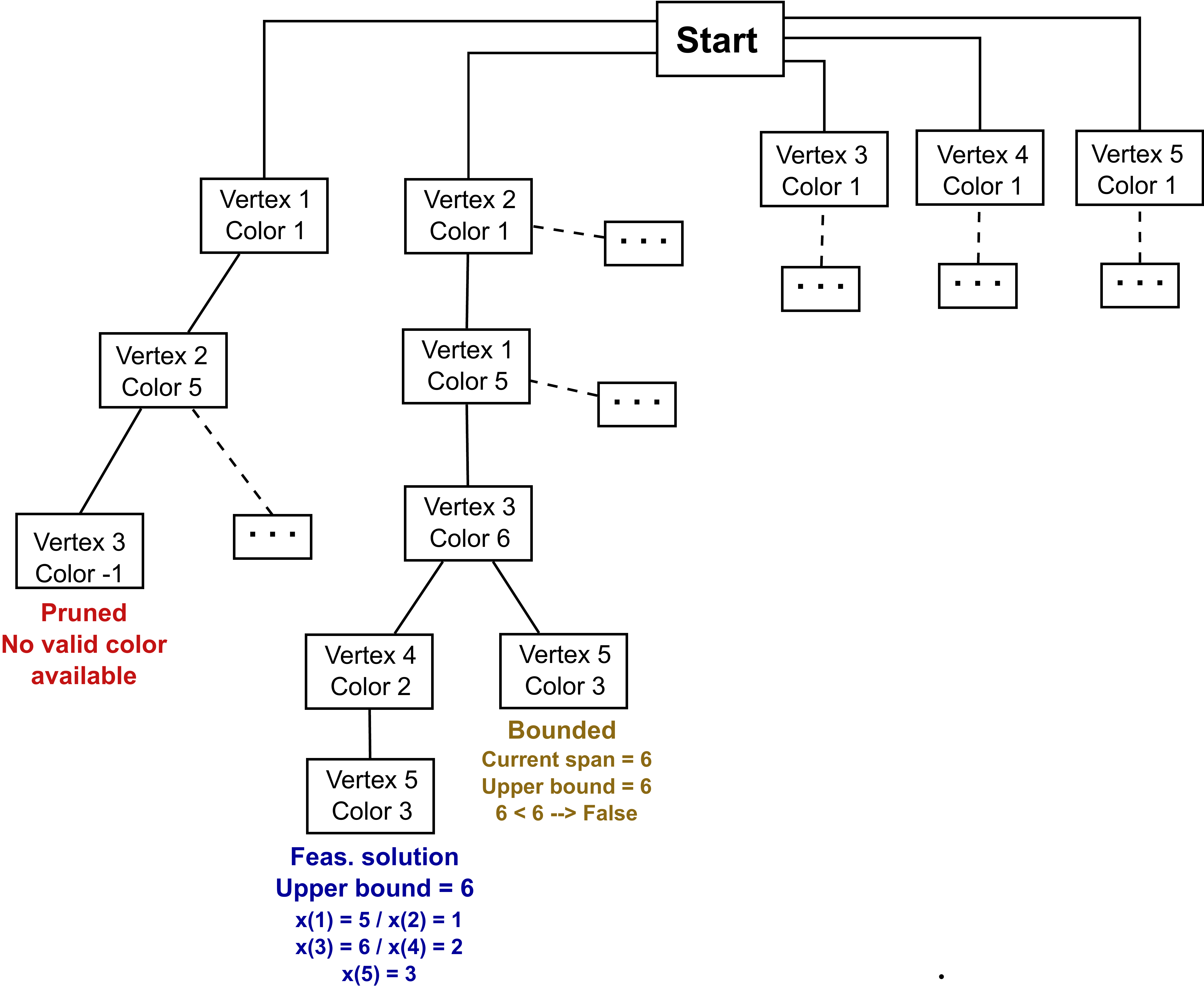}
\caption{Partial enumeration of solutions starting from vertex 2 for the MinEQ-CDGP instance defined by Figures \ref{fig:exNetDist} and \ref{fig:spheres} using BPB-Select, where a color is determined using a system of absolute value expressions (equalities or inequalities).}
\label{fig:enum-Select}
\end{figure}

\vspace{0.5cm}

\noindent{\textbf{Feasibility checking}} \\

When building a partial solution we must verify if it is feasible when not all distances are taken into account
at the same time, especially on BPB-Prev. We used a similar feasibility test to the Direct Distance Feasibility (DDF)
used on the BP algorithm for the DMDGP.

Let $i$ be the vertex that has just been colored. Then we must check, for each neighbor $j$ that has already been
colored, if the condition $|x(i)-x(j)| \ge d_{i, j}$ (if $f((i, j)) = 0$) or $|x(i)-x(j)| = d_{i, j}$ (if $f((i, j)) = 1$).
This test can be seen as a variation of DDF setting $\epsilon$ to zero and allowing inequalities in the test.
\revisado{We denote this procedure as Direct Distance Coloring Feasibility (DDCF) and its pseudocode
is given} in Algorithm \ref{alg:ddf}. \\

\begin{algorithm}
\caption{\ \revisado{Direct Distance Coloring Feasibility (DDCF)} check}
\label{alg:ddf}
\begin{spacing}{0.9}
%\small
\begin{algorithmic}[1]
\Require{graph $G$ (with set $V$ of vertices and set $E$ of edges), 
\revisado{problem type $t$ (MinGEQ-CDGP or MinEQ-CDGP)}, 
matrix $d$ of distances for each edge,
current coloring $x$ and vertex $i$.}
\Function{\revisado{DDCF-Check}}{$G, d, f, x, i$}
	\For{\textbf{each} neighbor $k$ of $i$}
		\If{$k$ is colored}
			\If{\revisado{$t$ = MinGEQ-CDGP}} \Comment Inequality constraint
				\If{$|x(k)-x(i)| > d_{i,j}$}
					\State \Return \textbf{false}
				\EndIf
			\Else \Comment Equality constraint
				\If{$|x(k)-x(i)| \neq d_{i,j}$}
					\State \Return \textbf{false}
				\EndIf
			\EndIf
		\EndIf
	\EndFor
	\State \Return \textbf{true}
\EndFunction
\end{algorithmic}
\end{spacing}
\end{algorithm}

%\vspace{0.5cm}

We note that when selecting a color using the first criterion (only taking into account the previously
colored vertex) the feasibility check can be made at each colored vertex or only when all vertices have been colored (which
will require that the function \textsc{DDF-Check}() is called for each vertex). \revisado{Each check (for only one vertex) runs in $O(|V|)$ time, and if the entire coloring is checked (that is, for all vertices), it runs in
$O(|V|^2)$ time}.
We also note that, when using
the second criterion (using a system of absolute value (in)equalities), the feasibility check can be skipped, since
the color that it returns is always feasible.

The combination of these selection criteria and the corresponding feasibility checks result in three
possible BPB algorithms, which are summarized in Table \ref{tbl:bpb}. \\

\begin{table} [H]
	\centering
	\caption{Summary of branch-prune-and-bound methods.}
	\label{tbl:bpb}
	\scalebox{0.72}{
		\begin{tabular}{ >{\centering}m{4cm}>{\centering}m{1.5cm} 
		c >{\centering}m{5.8cm}>{\centering}m{3cm} 
		c >{\centering}m{3cm}>{\centering}m{3cm}}
			\hline
			\multirow{2}{4cm}{\centering\textbf{Algorithm}} & 
			\multirow{2}{1.5cm}{\centering\textbf{Color set size}} & &
			\multicolumn{2}{c}{\textbf{Color selection of a vertex}} & &
			\multicolumn{2}{c}{\textbf{Feasibility checking}} \tabularnewline
			\cline{4-5} \cline{7-8} 
			& & & \textbf{Strategy} &	\textbf{Time complexity} & &
			\textbf{When} &	\textbf{Time complexity} \tabularnewline
			\hline
			BPB-Prev - Previous neighbor & 2 & &
			\parbox{6cm}{\centering\vspace*{2mm}$x(i) = x(j) + d_{i, j}$ or\\ 
			$x(i) = x(j) + d_{i, j}$\\ (if $x(i) > x(j) + d_{i, j}$)\vspace*{2mm}} & $O(1)$ & &
			At each colored vertex & $O(|V|)$ for each vertex \tabularnewline
			BPB-Prev-CheckFull - Alternate previous neighbor & 2 & &
			\parbox{6cm}{\centering\vspace*{1mm}$x(i) = x(j) + d_{i, j}$ or\\ 
			$x(i) = x(j) + d_{i, j}$\\ (if $x(i) > x(j) + d_{i, j}$)\vspace*{2mm}} & $O(1)$ & &
			Only when all vertices are colored & $O(|V|^2)$ for entire coloring\tabularnewline
			BPB-Select - System of all neighbors & 1 & &
			\parbox{6.1cm}{\centering\vspace*{2mm}$x(i) = \min\{k \in [1, UB] : 
			\forall (i,j) \in E\ |x(j) - k|\ = \text{(or $\ge$)}\ d_{i, j}
			\}$,\vspace*{2mm}} & $O(ub)$ & &
			Not needed & -\tabularnewline			
			%\hline
			\hline
		\end{tabular}}
	\end{table}

\vspace{0.5cm}

\section{Computational experiments}
\label{sec:compExp}

In order to analyze the behavior of the proposed distance geometry coloring problems and the branch-prune-and-bound algorithms, we made two main sets of experiments: the first one involved generating many random graphs with different numbers of vertices according to some configurations and counting how many include even or odd cycles (while the rest are trees), since some of the properties of distance geometry coloring are related to these types of graphs.

All algorithms used in these experiments were implemented in C language (compiled with \texttt{gcc} 4.8.4 using
options \texttt{-Ofast -march=native -mtune=native}) and executed on a computer 
equipped with an Intel Core i7-3770 (3,4GHz), 8GB of memory and Linux Mint 17 operating system.
We describe each set of experiments below.

\subsection{Counting members of graph classes in random instances}

\revisado{Using Theorems \ref{thm:EQ-CDGP-Unif} and \ref{thm:EQ-CDGParbdist}, we have information about some types of graphs which
always have feasible embeddings for EQ-CDGP and EQ-CDGP-Unif. Based on this, we generated a large amount of random graphs with different number of vertices and counted how many were cyclic (and based on that, how many there were for each
possibility of having even or odd cycles) and how many were trees.}

\revisado{Each random graph always starts as a random spanning tree, that is, a connected undirected 
graph $G = (V, E)$, where $|E| = 
|V|-1$. To generate this initial tree, we used a random walk algorithm proposed independently by 
\cite{broder:1989} and \cite{aldous:1990}. The procedure works by using a set $V^*$
of the vertices outside the tree and a set $W$ of edges of the spanning tree.
Then, whenever the random walk reaches a vertex $j$ outside the tree, the edge $(i, j)$ is added to $E$
and $j$ is removed from $V^*$. This continues until $V^* = \varnothing$. We note that this amounts to making
a random walk in a complete graph of $|V|$ vertices and it generates trees in a uniform manner, that is,
for all possible spanning trees of a given complete graph, each one has the same probability of being
generated by the algorithm.}

%\begin{algorithm}
%\caption{\ \revisado{Aldous-Broder Algorithm}}
%\label{alg:ab}
%\begin{spacing}{0.9}
%%\small
%\begin{algorithmic}[1]
%\Require{number of desired vertices $n$.}
%\Function{\revisado{Aldous-Broder}}{$n$}
%	\State $V \gets \{1, 2, \dots, n\}$
%	\State $E \gets \varnothing$
%	\State $i \gets 1$
%	\State $V^* \gets V - \{1\}$
%	\While{$V^* \neq \varnothing$}
%		\State $j \gets$ random vertex from $V$
%		\If{$j \in V^*$}
%			\State $E \gets E \cup \{(i, j)\}$
%			\State $V^* \gets V^* - \{j\}$
%		\EndIf
%	\EndWhile
%	\State \Return $(V, E)$
%\EndFunction
%\end{algorithmic}
%\end{spacing}
%\end{algorithm}

\revisado{After the initial tree is generated, we add random new edges to it until the graph has the desired number
of edges. This parameter is also randomly set, sampled from interval $\left[|V|-1, \frac{|V| (|V|-1)}{2}\right]$. This
interval ensures that the generated graph is always connected and is, at least, a tree and, at most, a
complete graph.}

\revisado{In Table \ref{tbl:random}, we outline statistics obtained from using the described procedure to generate 1,000,000 (one
million) random graphs for each $|V| \in \{50, 100, 150, 200, 250, 300, 350, 400, 450, 500\}$. As we can observe,
most of the graphs (more than 99\%) generated have odd cycles, which translates into a very small set
of possible EQ-CDGP-Unif instances with feasible embeddings for this configuration of random graphs.
By increasing the number of vertices, more possibilities for generating edges appear, but the number
of possible connections which will lead to trees or graphs with even cycles is very small. In fact,
we can deduce that this configuration generated very few bipartite graphs. For EQ-CDGP (with arbitrary
distances), the space of instances with guaranteed feasible embeddings is even smaller, since only
trees are certain to have them. However, as shown in Section \ref{sec:distColProp}, 
odd and even cycles can have embeddings depending on how the edges are weighted.}

\revisado{In Figure \ref{fig:vXe}, we can observe the growth of the average number of edges between all generated
graphs for each number of vertices. Since the number of edges in a graph is proportional to the
square of the number of edges (since $\frac{|V| (|V|-1)}{2}\ \in\ O(|V|^2)$, the curve follows
a similar pattern, being a half parabola.}

\begin{table}
\caption{Number of random graphs with even, odd or no cycles (trees) and bipartite graphs generated for each number of vertices. For each size,
1,000,000 graphs were generated.}
\centering
\resizebox{\textwidth}{!}{%
\begin{tabular}{>{\centering}m{1cm} >{\centering}m{1.6cm} >{\centering}m{2cm} >{\centering}m{3cm} >{\centering}m{3cm} 
>{\centering}m{2cm} >{\centering}m{3cm} >{\centering}m{2.3cm}}
\hline
\textbf{\boldmath$|V|$} & {\centering \bf Average \boldmath$|E|$} & {\centering \bf Average Density} & 
\textbf{\# Graphs with Odd Cycles} & \textbf{\# Graphs with Even Cycles} & \textbf{\# Trees} &
\textbf{\# Bipartite Graphs} & \textbf{CPU Time (sec)}\tabularnewline
\hline
50 & 637.56 & 0.5205 & 998309 & 854 & 837 & 1691 & 257.07 \tabularnewline
100 & 2523.52 & 0.5098 & 999553 & 238 & 209 & 447 & 753.25 \tabularnewline
150 & 5656.64 & 0.5062 & 999832 & 74 & 94 & 168 & 1808.25 \tabularnewline
200 & 10059.56 & 0.5055 & 999910 & 45 & 45 & 90 & 3403.02 \tabularnewline
250 & 15675.94 & 0.5036 & 999926 & 41 & 33 & 74 & 4553.07 \tabularnewline
300 & 22586.52 & 0.5036 & 999958 & 18 & 24 & 42 & 6764.28 \tabularnewline
350 & 30688.21 & 0.5025 & 999975 & 13 & 12 & 25 & 10042.43 \tabularnewline
400 & 40120.76 & 0.5028 & 999975 & 14 & 11 & 25 & 11886.32 \tabularnewline
450 & 50678.60 & 0.5016 & 999971 & 15 & 14 & 29 & 14415.33 \tabularnewline
500 & 62628.32 & 0.5020 & 999988 & 6 & 6 & 12 & 23332.64 \tabularnewline
\hline
\end{tabular}}
\label{tbl:random}
\end{table}

\begin{figure}
\centering
\includegraphics[scale=0.8]{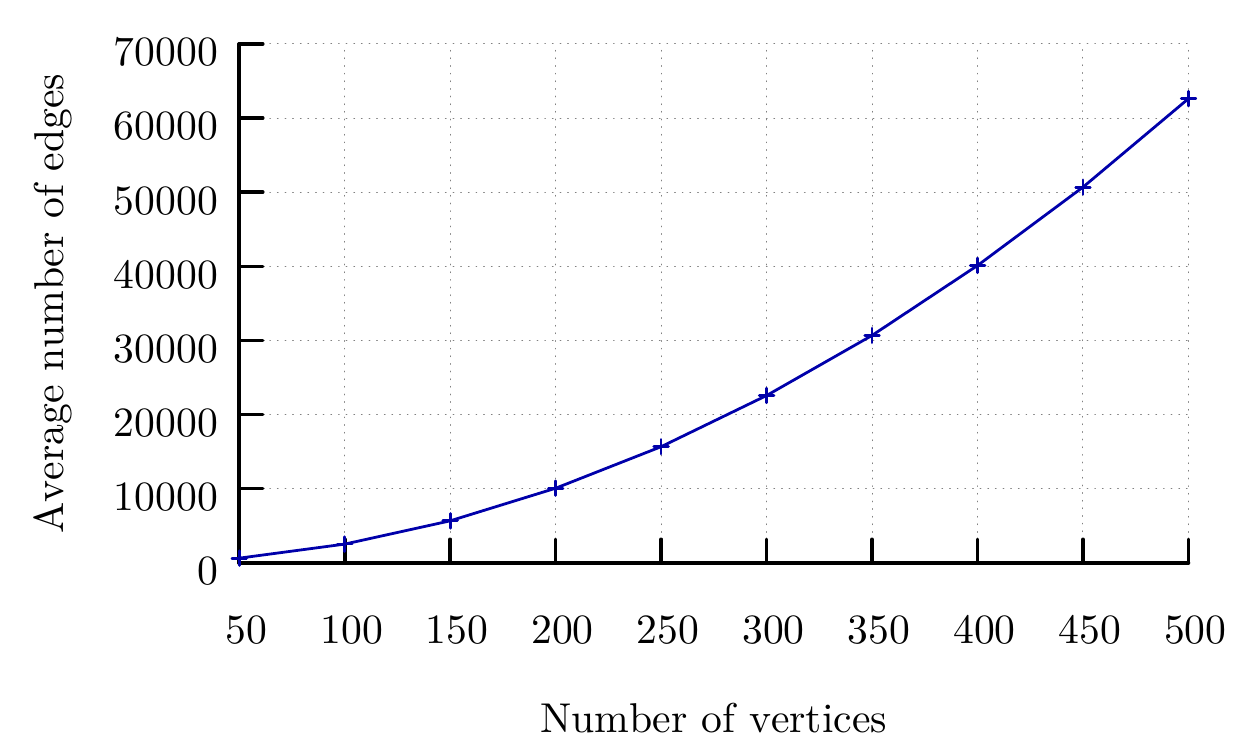}
\caption{Growth of the average number of edges generated when the number of vertices increases.}
\label{fig:vXe}
\end{figure}

%\begin{figure}
%	\centering
%	\subfigure[][\mbox{$|V| \in [50, 500]$}]{\includegraphics[scale=0.8]{images/bips5000-1.pdf}}
%	\hspace*{1cm}
%	\subfigure[][\mbox{$|V| \in [200, 500]$ (for scale)}]{\includegraphics[scale=0.8]{images/bips5000-2-scale.pdf}}
%	\caption{Average number of bipartite graphs generated for each 5000 random graphs of each number of vertices.}
%	\label{fig:bips-5000}
%\end{figure}

\begin{figure}
	\centering
	\subfigure[][\mbox{$|V| \in [50, 500]$}]{\includegraphics[scale=0.8]{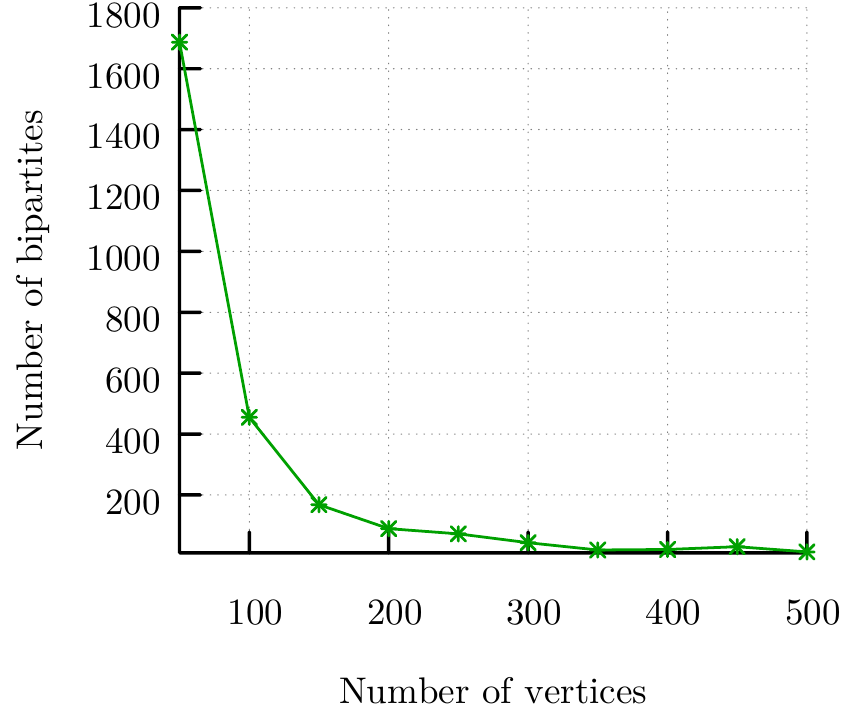}}
	\hspace*{1cm}
	\subfigure[][\mbox{$|V| \in [200, 500]$ (for scale)}]{\includegraphics[scale=0.8]{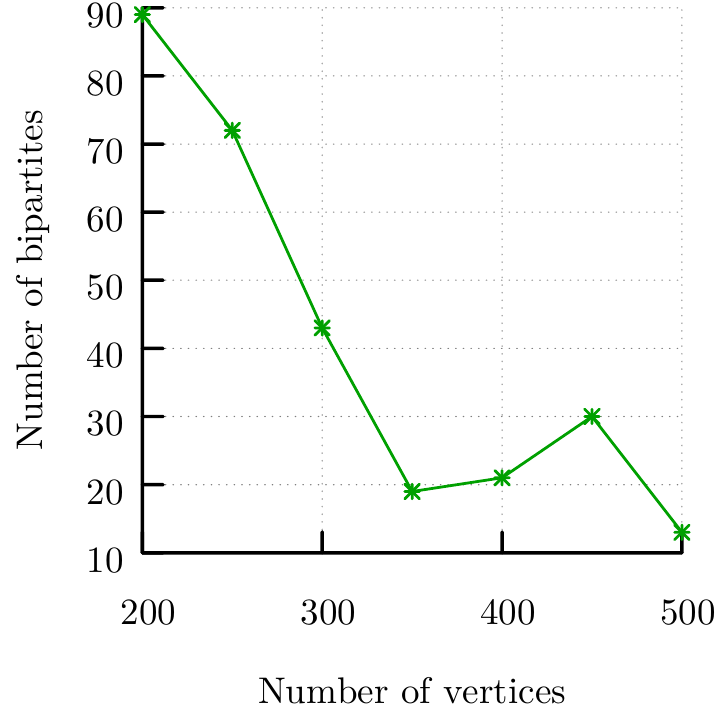}}
	\caption{Total number of bipartite graphs generated from 1,000,000 random graphs of each number of vertices.}
	\label{fig:bips-Total}
\end{figure}

\subsection{Results for branch-prune-and-bound algorithms}

\begin{table}[]
\centering
\caption{Results for BP algorithm (decision/search) applied to EQ-CDGP instances - 4, 5, 6, 7 and 8 vertices.}
\label{tbl:BP-EQ-4to8}
\resizebox{\textwidth}{!}{%
\begin{tabular}{ccc|cccc|cccc|cccc}
\cline{4-15}
                     &           &      & \multicolumn{4}{c}{\textbf{BPB-Prev}}                 & \multicolumn{4}{|c}{\textbf{BPB-Prev-FeasCheckFull}}   & \multicolumn{4}{|c}{\textbf{BPB-Select}}           \\ \hline
\textbf{\boldmath$|$V$|$}                  & \textbf{Type}      & \textbf{Inst} 
& \textbf{Span} & \textbf{\# Prunes} & \textbf{\# Nodes} & \textbf{CPU Time (s)} 
& \textbf{Span} & \textbf{\# Prunes} & \textbf{\# Nodes} & \textbf{CPU Time (s)}
& \textbf{Span} & \textbf{\# Prunes} & \textbf{\# Nodes} & \textbf{CPU Time (s)} \\ \hline
\multirow{12}{*}{4}  & OddCycle  & 1    & Infeasible & 28          & 36          & 0.000        & Infeasible & 36          & 54          & 0.000        & Infeasible & 16        & 34        & 0.000        \\
                     & OddCycle  & 2    & Infeasible & 21          & 32          & 0.000        & Infeasible & 25          & 42          & 0.000        & Infeasible & 12        & 30        & 0.000        \\
                     & OddCycle  & 3    & Infeasible & 21          & 32          & 0.000        & Infeasible & 22          & 41          & 0.000        & Infeasible & 12        & 30        & 0.000        \\
                     & OddCycle  & 4    & Infeasible & 19          & 32          & 0.000        & Infeasible & 21          & 41          & 0.000        & Infeasible & 12        & 30        & 0.000        \\ \cline{2-15} 
                     & EvenCycle & 1    & Infeasible & 20          & 32          & 0.000        & Infeasible & 20          & 32          & 0.000        & Infeasible & 8         & 28        & 0.000        \\
                     & EvenCycle & 2    & Infeasible & 20          & 32          & 0.000        & Infeasible & 20          & 32          & 0.000        & Infeasible & 8         & 28        & 0.000        \\
                     & EvenCycle & 3    & 24         & 0           & 4           & 0.000        & 24         & 0           & 4           & 0.000        & 24         & 0         & 4         & 0.000        \\
                     & EvenCycle & 4    & 19         & 0           & 4           & 0.000        & 19         & 0           & 4           & 0.000        & 19         & 0         & 4         & 0.000        \\ \cline{2-15} 
                     & Tree      & 1    & 21         & 0           & 4           & 0.000        & 21         & 0           & 4           & 0.000        & 14         & 0         & 4         & 0.000        \\
                     & Tree      & 2    & 34         & 0           & 4           & 0.000        & 34         & 0           & 4           & 0.000        & 14         & 0         & 4         & 0.000        \\
                     & Tree      & 3    & 29         & 0           & 4           & 0.000        & 29         & 0           & 4           & 0.000        & 20         & 0         & 4         & 0.000        \\
                     & Tree      & 4    & 41         & 0           & 4           & 0.000        & 41         & 0           & 4           & 0.000        & 21         & 0         & 4         & 0.000        \\ \hline
\multirow{12}{*}{5}  & OddCycle  & 1    & 38         & 0           & 5           & 0.000        & 38         & 0           & 5           & 0.000        & 38         & 1         & 7         & 0.000        \\
                     & OddCycle  & 2    & Infeasible & 38          & 58          & 0.000        & Infeasible & 56          & 88          & 0.000        & Infeasible & 20        & 54        & 0.000        \\
                     & OddCycle  & 3    & 20         & 1           & 5           & 0.000        & 20         & 1           & 6           & 0.000        & 20         & 0         & 5         & 0.000        \\
                     & OddCycle  & 4    & Infeasible & 61          & 69          & 0.000        & Infeasible & 106         & 175         & 0.000        & Infeasible & 30        & 62        & 0.000        \\ \cline{2-15} 
                     & EvenCycle & 1    & 44         & 0           & 5           & 0.000        & 44         & 0           & 5           & 0.000        & 32         & 1         & 7         & 0.000        \\
                     & EvenCycle & 2    & 36         & 0           & 5           & 0.000        & 36         & 0           & 5           & 0.000        & 24         & 0         & 5         & 0.000        \\
                     & EvenCycle & 3    & Infeasible & 54          & 77          & 0.000        & Infeasible & 66          & 97          & 0.000        & Infeasible & 20        & 60        & 0.000        \\
                     & EvenCycle & 4    & Infeasible & 51          & 76          & 0.000        & Infeasible & 62          & 96          & 0.000        & Infeasible & 20        & 60        & 0.000        \\ \cline{2-15} 
                     & Tree      & 1    & 20         & 0           & 5           & 0.000        & 20         & 0           & 5           & 0.000        & 17         & 0         & 5         & 0.000        \\
                     & Tree      & 2    & 19         & 0           & 5           & 0.000        & 19         & 0           & 5           & 0.000        & 19         & 0         & 5         & 0.000        \\
                     & Tree      & 3    & 23         & 0           & 5           & 0.000        & 23         & 0           & 5           & 0.000        & 20         & 0         & 5         & 0.000        \\
                     & Tree      & 4    & 39         & 0           & 5           & 0.000        & 39         & 0           & 5           & 0.000        & 20         & 1         & 7         & 0.000        \\ \hline
\multirow{12}{*}{6}  & OddCycle  & 1    & Infeasible & 149         & 161         & 0.000        & Infeasible & 378         & 533         & 0.001        & Infeasible & 45        & 105       & 0.000        \\
                     & OddCycle  & 2    & Infeasible & 176         & 168         & 0.000        & Infeasible & 887         & 1114        & 0.001        & Infeasible & 76        & 128       & 0.000        \\
                     & OddCycle  & 3    & Infeasible & 167         & 161         & 0.000        & Infeasible & 374         & 465         & 0.000        & Infeasible & 48        & 105       & 0.000        \\
                     & OddCycle  & 4    & Infeasible & 126         & 135         & 0.000        & Infeasible & 504         & 624         & 0.001        & Infeasible & 43        & 95        & 0.000        \\ \cline{2-15} 
                     & EvenCycle & 1    & Infeasible & 100         & 131         & 0.000        & Infeasible & 122         & 187         & 0.000        & Infeasible & 30        & 97        & 0.000        \\
                     & EvenCycle & 2    & 46         & 0           & 6           & 0.000        & 46         & 0           & 6           & 0.000        & 21         & 0         & 6         & 0.000        \\
                     & EvenCycle & 3    & Infeasible & 155         & 223         & 0.000        & Infeasible & 242         & 381         & 0.000        & Infeasible & 60        & 160       & 0.000        \\
                     & EvenCycle & 4    & Infeasible & 84          & 137         & 0.000        & Infeasible & 120         & 203         & 0.000        & Infeasible & 30        & 100       & 0.000        \\ \cline{2-15} 
                     & Tree      & 1    & 16         & 0           & 6           & 0.000        & 16         & 0           & 6           & 0.000        & 13         & 0         & 6         & 0.000        \\
                     & Tree      & 2    & 54         & 0           & 6           & 0.000        & 54         & 0           & 6           & 0.000        & 19         & 0         & 6         & 0.000        \\
                     & Tree      & 3    & 34         & 0           & 6           & 0.000        & 34         & 0           & 6           & 0.000        & 27         & 5         & 17        & 0.000        \\
                     & Tree      & 4    & 39         & 0           & 6           & 0.000        & 39         & 0           & 6           & 0.000        & 39         & 0         & 6         & 0.000        \\ \hline
\multirow{12}{*}{7}  & OddCycle  & 1    & Infeasible & 191         & 196         & 0.000        & Infeasible & 2234        & 2880        & 0.003        & Infeasible & 66        & 139       & 0.000        \\
                     & OddCycle  & 2    & Infeasible & 360         & 307         & 0.000        & Infeasible & 5398        & 6875        & 0.007        & Infeasible & 148       & 229       & 0.000        \\
                     & OddCycle  & 3    & Infeasible & 383         & 323         & 0.000        & Infeasible & 5360        & 6506        & 0.006        & Infeasible & 138       & 219       & 0.000        \\
                     & OddCycle  & 4    & Infeasible & 378         & 310         & 0.000        & Infeasible & 5391        & 6470        & 0.006        & Infeasible & 144       & 227       & 0.000        \\ \cline{2-15} 
                     & EvenCycle & 1    & Infeasible & 263         & 361         & 0.000        & Infeasible & 429         & 686         & 0.001        & Infeasible & 91        & 246       & 0.000        \\
                     & EvenCycle & 2    & 36         & 151         & 193         & 0.000        & 36         & 192         & 243         & 0.000        & 23         & 26        & 105       & 0.000        \\
                     & EvenCycle & 3    & Infeasible & 250         & 283         & 0.000        & Infeasible & 1093        & 1422        & 0.001        & Infeasible & 70        & 178       & 0.000        \\
                     & EvenCycle & 4    & Infeasible & 293         & 382         & 0.000        & Infeasible & 511         & 745         & 0.001        & Infeasible & 100       & 259       & 0.000        \\ \cline{2-15} 
                     & Tree      & 1    & 34         & 0           & 7           & 0.000        & 34         & 0           & 7           & 0.000        & 19         & 0         & 7         & 0.000        \\
                     & Tree      & 2    & 33         & 0           & 7           & 0.000        & 33         & 0           & 7           & 0.000        & 14         & 0         & 7         & 0.000        \\
                     & Tree      & 3    & 27         & 0           & 7           & 0.000        & 27         & 0           & 7           & 0.000        & 27         & 2         & 11        & 0.000        \\
                     & Tree      & 4    & 35         & 0           & 7           & 0.000        & 35         & 0           & 7           & 0.000        & 35         & 1         & 10        & 0.000        \\ \hline
\multirow{12}{*}{8}  & OddCycle  & 1    & Infeasible & 650         & 505         & 0.001        & Infeasible & 34796       & 41815       & 0.041        & Infeasible & 220       & 336       & 0.000        \\
                     & OddCycle  & 2    & Infeasible & 896         & 1202        & 0.001        & Infeasible & 1750        & 2191        & 0.002        & Infeasible & 145       & 438       & 0.000        \\
                     & OddCycle  & 3    & Infeasible & 696         & 652         & 0.001        & Infeasible & 3192        & 4052        & 0.004        & Infeasible & 99        & 251       & 0.000        \\
                     & OddCycle  & 4    & Infeasible & 1092        & 1075        & 0.001        & Infeasible & 13940       & 18078       & 0.018        & Infeasible & 311       & 563       & 0.001        \\ \cline{2-15} 
                     & EvenCycle & 1    & 51         & 0           & 8           & 0.000        & 51         & 0           & 8           & 0.000        & 28         & 1         & 10        & 0.000        \\
                     & EvenCycle & 2    & Infeasible & 1005        & 1072        & 0.001        & Infeasible & 5410        & 6544        & 0.007        & Infeasible & 230       & 511       & 0.000        \\
                     & EvenCycle & 3    & Infeasible & 1474        & 1934        & 0.002        & Infeasible & 3346        & 4369        & 0.005        & Infeasible & 352       & 883       & 0.001        \\
                     & EvenCycle & 4    & 40         & 0           & 8           & 0.000        & 40         & 0           & 8           & 0.000        & 18         & 4         & 19        & 0.000        \\ \cline{2-15} 
                     & Tree      & 1    & 45         & 0           & 8           & 0.000        & 45         & 0           & 8           & 0.000        & 34         & 0         & 8         & 0.000        \\
                     & Tree      & 2    & 47         & 0           & 8           & 0.000        & 47         & 0           & 8           & 0.000        & 29         & 0         & 8         & 0.000        \\
                     & Tree      & 3    & 43         & 0           & 8           & 0.000        & 43         & 0           & 8           & 0.000        & 22         & 0         & 8         & 0.000        \\
                     & Tree      & 4    & 71         & 0           & 8           & 0.000        & 71         & 0           & 8           & 0.000        & 20         & 1         & 10        & 0.000        \\ \hline
\end{tabular}
}
\end{table}

\begin{table}[]
\centering
\caption{Results for BP algorithm (decision/search) applied to EQ-CDGP instances - 9, 10, 12, 14 and 16 vertices.}
\label{tbl:BP-EQ-9to16}
\resizebox{\textwidth}{!}{%
\begin{tabular}{ccc|cccc|cccc|cccc}
\cline{4-15}
                     &           &      & \multicolumn{4}{c}{\textbf{BPB-Prev}}                 & \multicolumn{4}{|c}{\textbf{BPB-Prev-FeasCheckFull}}   & \multicolumn{4}{|c}{\textbf{BPB-Select}}           \\ \hline
\textbf{\boldmath$|$V$|$}                  & \textbf{Type}      & \textbf{Inst} 
& \textbf{Span} & \textbf{\# Prunes} & \textbf{\# Nodes} & \textbf{CPU Time (s)} 
& \textbf{Span} & \textbf{\# Prunes} & \textbf{\# Nodes} & \textbf{CPU Time (s)}
& \textbf{Span} & \textbf{\# Prunes} & \textbf{\# Nodes} & \textbf{CPU Time (s)} \\ \hline
\multirow{12}{*}{9}  & OddCycle  & 1    & Infeasible & 1638        & 1537        & 0.002        & Infeasible & 46903       & 60000       & 0.065        & Infeasible & 434       & 777       & 0.001        \\
                     & OddCycle  & 2    & Infeasible & 1633        & 1219        & 0.001        & Infeasible & 328822      & 400892      & 0.397        & Infeasible & 456       & 650       & 0.001        \\
                     & OddCycle  & 3    & Infeasible & 1165        & 923         & 0.001        & Infeasible & 385150      & 467221      & 0.456        & Infeasible & 452       & 620       & 0.000        \\
                     & OddCycle  & 4    & Infeasible & 1417        & 1801        & 0.002        & Infeasible & 2479        & 3574        & 0.003        & Infeasible & 297       & 801       & 0.001        \\ \cline{2-15} 
                     & EvenCycle & 1    & 61         & 0           & 9           & 0.000        & 61         & 0           & 9           & 0.000        & 34         & 0         & 9         & 0.000        \\
                     & EvenCycle & 2    & 26         & 0           & 9           & 0.000        & 26         & 0           & 9           & 0.000        & 20         & 7         & 24        & 0.000        \\
                     & EvenCycle & 3    & Infeasible & 2011        & 2542        & 0.003        & Infeasible & 3414        & 4392        & 0.005        & Infeasible & 226       & 788       & 0.000        \\
                     & EvenCycle & 4    & Infeasible & 946         & 1142        & 0.001        & Infeasible & 1347        & 1785        & 0.002        & Infeasible & 88        & 357       & 0.000        \\ \cline{2-15} 
                     & Tree      & 1    & 55         & 0           & 9           & 0.000        & 55         & 0           & 9           & 0.000        & 38         & 10        & 30        & 0.000        \\
                     & Tree      & 2    & 68         & 0           & 9           & 0.000        & 68         & 0           & 9           & 0.000        & 27         & 0         & 9         & 0.000        \\
                     & Tree      & 3    & 88         & 0           & 9           & 0.000        & 88         & 0           & 9           & 0.000        & 29         & 0         & 9         & 0.000        \\
                     & Tree      & 4    & 45         & 0           & 9           & 0.000        & 45         & 0           & 9           & 0.000        & 21         & 0         & 9         & 0.000        \\ \hline
\multirow{12}{*}{10} & OddCycle  & 1    & Infeasible & 33183       & 42018       & 0.042        & Infeasible & 155000      & 173304      & 0.168        & Infeasible & 7939      & 15821     & 0.011        \\
                     & OddCycle  & 2    & Infeasible & 2619        & 1974        & 0.002        & Infeasible & 2108251     & 2797473     & 2.814        & Infeasible & 792       & 1127      & 0.001        \\
                     & OddCycle  & 3    & Infeasible & 3529        & 2963        & 0.003        & Infeasible & 316389      & 400670      & 0.367        & Infeasible & 710       & 1184      & 0.001        \\
                     & OddCycle  & 4    & Infeasible & 4379        & 3333        & 0.004        & Infeasible & 1481970     & 1818952     & 1.868        & Infeasible & 1030      & 1441      & 0.001        \\ \cline{2-15} 
                     & EvenCycle & 1    & Infeasible & 3500        & 3114        & 0.003        & Infeasible & 157665      & 199713      & 0.174        & Infeasible & 520       & 1036      & 0.001        \\
                     & EvenCycle & 2    & Infeasible & 46618       & 48740       & 0.045        & Infeasible & 160032      & 193409      & 0.224        & Infeasible & 4718      & 9984      & 0.006        \\
                     & EvenCycle & 3    & Infeasible & 70914       & 86532       & 0.086        & Infeasible & 254208      & 304707      & 0.249        & Infeasible & 4662      & 10887     & 0.007        \\
                     & EvenCycle & 4    & 72         & 1           & 10          & 0.000        & 72         & 1           & 10          & 0.000        & 38         & 0         & 10        & 0.000        \\ \cline{2-15} 
                     & Tree      & 1    & 59         & 0           & 10          & 0.000        & 59         & 0           & 10          & 0.000        & 27         & 4         & 18        & 0.000        \\
                     & Tree      & 2    & 40         & 0           & 10          & 0.000        & 40         & 0           & 10          & 0.000        & 31         & 1         & 12        & 0.000        \\
                     & Tree      & 3    & 49         & 0           & 10          & 0.000        & 49         & 0           & 10          & 0.000        & 33         & 4         & 18        & 0.000        \\
                     & Tree      & 4    & 46         & 0           & 10          & 0.000        & 46         & 0           & 10          & 0.000        & 29         & 3         & 16        & 0.000        \\ \hline
\multirow{12}{*}{12} & OddCycle  & 1    & Infeasible & 108359      & 103428      & 0.115        & Infeasible & 965368      & 1140595     & 1.284        & Infeasible & 9947      & 18283     & 0.012        \\
                     & OddCycle  & 2    & Infeasible & 24433       & 23227       & 0.026        & Infeasible & 1559180     & 1807872     & 1.927        & Infeasible & 1578      & 3718      & 0.003        \\
                     & OddCycle  & 3    & Infeasible & 7440        & 5294        & 0.006        & Infeasible & 36724085    & 44121832    & 41.883       & Infeasible & 1426      & 2058      & 0.002        \\
                     & OddCycle  & 4    & Infeasible & 8467        & 5648        & 0.006        & Infeasible & 544947529   & 645413373 & 390.399     & Infeasible   & 2293       & 2876         & 0.002                 \\ \cline{2-15} 
                     & EvenCycle & 1    & 49         & 0           & 12          & 0.000        & 49         & 0           & 12          & 0.000        & 22         & 42        & 146       & 0.000        \\
                     & EvenCycle & 2    & Infeasible & 13801       & 16577       & 0.017        & Infeasible & 53860       & 62497       & 0.073        & Infeasible & 954       & 2799      & 0.002        \\
                     & EvenCycle & 3    & 41         & 0           & 12          & 0.000        & 41         & 0           & 12          & 0.000        & 35         & 16        & 51        & 0.000        \\
                     & EvenCycle & 4    & Infeasible & 18820       & 19828       & 0.021        & Infeasible & 381806      & 495804      & 0.393        & Infeasible & 6572      & 13539     & 0.009        \\ \cline{2-15} 
                     & Tree      & 1    & 38         & 0           & 12          & 0.000        & 38         & 0           & 12          & 0.000        & 27         & 58        & 139       & 0.000        \\
                     & Tree      & 2    & 50         & 0           & 12          & 0.000        & 50         & 0           & 12          & 0.000        & 22         & 0         & 12        & 0.000        \\
                     & Tree      & 3    & 48         & 0           & 12          & 0.000        & 48         & 0           & 12          & 0.000        & 34         & 440       & 1016      & 0.001        \\
                     & Tree      & 4    & 52         & 0           & 12          & 0.000        & 52         & 0           & 12          & 0.000        & 26         & 2         & 16        & 0.000        \\ \hline
\multirow{12}{*}{14} & OddCycle  & 1    & Infeasible & 2725363     & 2744206     & 2.820        & Infeasible & 104517592   & 118562122 & 101.605     & Infeasible   & 119828     & 218219         & 0.146       \\
                     & OddCycle  & 2    & Infeasible & 34217       & 23821       & 0.019        & Infeasible & 3407171273  & 4097250986  & 3314.219     & Infeasible & 5217      & 6949      & 0.005        \\
                     & OddCycle  & 3    & Infeasible & 25749       & 16751       & 0.018        & Infeasible & 11200102605 & 14120166774 & 10800.000    & Infeasible & 4890      & 6169      & 0.005        \\
                     & OddCycle  & 4    & Infeasible & 17520       & 11848       & 0.013        & Infeasible & 4872771100  & 5921495589  & 5148.447     & Infeasible & 2844      & 3883      & 0.003        \\ \cline{2-15} 
                     & EvenCycle & 1    & Infeasible & 22438       & 20427       & 0.022        & Infeasible & 1726508     & 1934194 & 1.889       & Infeasible   & 286        & 955         & 0,001     \\
                     & EvenCycle & 2    & Infeasible & 8815580     & 9620137     & 9.640        & Infeasible & 38618944    & 42925314 & 27.755      & Infeasible   & 240185     &  480564         & 0.306 \\
                     & EvenCycle & 3    & Infeasible & 8022290     & 8013781     & 7.269        & Infeasible & 32873088    & 37749326    & 29.182       & Infeasible & 146774    & 329350    & 0.205        \\
                     & EvenCycle & 4    & Infeasible & 4979521     & 5359288     & 4.921        & Infeasible & 16804920    & 21169856    & 17.424       & Infeasible & 160477    & 354148    & 0.225        \\ \cline{2-15} 
                     & Tree      & 1    & 55         & 0           & 14          & 0.000        & 55         & 0           & 14          & 0.000        & 29         & 54        & 138       & 0.000        \\
                     & Tree      & 2    & 63         & 0           & 14          & 0.000        & 63         & 0           & 14          & 0.000        & 37         & 164       & 385       & 0.000        \\
                     & Tree      & 3    & 67         & 0           & 14          & 0.000        & 67         & 0           & 14          & 0.000        & 31         & 60        & 166       & 0.000        \\
                     & Tree      & 4    & 59         & 0           & 14          & 0.000        & 59         & 0           & 14          & 0.000        & 39         & 647       & 1432      & 0.001        \\ \hline
\multirow{12}{*}{16} & OddCycle  & 1    & Infeasible & 9674081     & 8970218     & 8.843        & Infeasible & 270950088   & 312103257   & 218.788      & Infeasible & 229364    & 464467    & 0.295        \\
                     & OddCycle  & 2    & Infeasible & 40827       & 26130       & 0.019        & Infeasible & 9212314467  & 13585438967 & 10800.000    & Infeasible & 10334     & 12037     & 0.010        \\
                     & OddCycle  & 3    & Infeasible & 64340       & 40739       & 0.031        & Infeasible & 12857613205 & 16238297924 & 10800.000    & Infeasible & 13417     & 15677     & 0.012        \\
                     & OddCycle  & 4    & Infeasible & 79312       & 71850       & 0.056        & Infeasible & 135939959   & 153954534   & 104.928      & Infeasible & 7508      & 14908     & 0.010        \\ \cline{2-15} 
                     & EvenCycle & 1    & 79         & 0           & 16          & 0.000        & 79         & 0           & 16          & 0.000        & 35         & 356       & 799       & 0.001        \\
                     & EvenCycle & 2    & Infeasible & 41169779    & 41680959    & 35.448       & Infeasible & 881601120   & 1035374397  & 712.612      & Infeasible & 10410162  & 19341115  & 12.607       \\
                     & EvenCycle & 3    & 72         & 0           & 16          & 0.000        & 72         & 0           & 16          & 0.000        & 29         & 1390      & 3171      & 0.002        \\
                     & EvenCycle & 4    & Infeasible & 61277800    & 71043803    & 55.371       & Infeasible & 250169040   & 273392839   & 196.395      & Infeasible & 787682    & 1851608   & 1.155        \\ \cline{2-15} 
                     & Tree      & 1    & 65         & 0           & 16          & 0.000        & 65         & 0           & 16          & 0.000        & 30         & 222       & 502       & 0.000        \\
                     & Tree      & 2    & 86         & 0           & 16          & 0.000        & 86         & 0           & 16          & 0.000        & 30         & 15        & 50        & 0.000        \\
                     & Tree      & 3    & 74         & 0           & 16          & 0.000        & 74         & 0           & 16          & 0.000        & 36         & 329       & 756       & 0.001        \\
                     & Tree      & 4    & 79         & 0           & 16          & 0.000        & 79         & 0           & 16          & 0.000        & 23         & 9         & 42        & 0.000        \\ \hline

\end{tabular}
}
\end{table}
                     
\begin{table}[]
\centering
\caption{Results for BP algorithm (decision/search) applied to EQ-CDGP instances - 18 and 20 vertices.}
\label{tbl:BP-EQ-18to20}
\resizebox{\textwidth}{!}{%
\begin{tabular}{ccc|cccc|cccc|cccc}
\cline{4-15}
                     &           &      & \multicolumn{4}{c}{\textbf{BPB-Prev}}                 & \multicolumn{4}{|c}{\textbf{BPB-Prev-FeasCheckFull}}   & \multicolumn{4}{|c}{\textbf{BPB-Select}}           \\ \hline
\textbf{\boldmath$|$V$|$}                  & \textbf{Type}      & \textbf{Inst} 
& \textbf{Span} & \textbf{\# Prunes} & \textbf{\# Nodes} & \textbf{CPU Time (s)} 
& \textbf{Span} & \textbf{\# Prunes} & \textbf{\# Nodes} & \textbf{CPU Time (s)}
& \textbf{Span} & \textbf{\# Prunes} & \textbf{\# Nodes} & \textbf{CPU Time (s)} \\ \hline
\multirow{12}{*}{18} & OddCycle  & 1    & Infeasible & 508400      & 411189      & 0.324        & Infeasible & 14249119873 & 14806963853 & 10800.000    & Infeasible & 18527     & 33062     & 0.023        \\
                     & OddCycle  & 2    & Infeasible & 1151865     & 809413      & 0.630        & Infeasible & 6511737586  & 7859547065  & 10800.000    & Infeasible & 29654     & 46174     & 0.032        \\
                     & OddCycle  & 3    & Infeasible & 164231      & 137272      & 0.093        & Infeasible & 2830789695  & 3168256401  & 3927.393     & Infeasible & 1411      & 3456      & 0.002        \\
                     & OddCycle  & 4    & Infeasible & 117988      & 72425       & 0.050        & Infeasible & 8764779022  & 9900318687  & 10800.000    & Infeasible & 25336     & 28579     & 0.022        \\ \cline{2-15} 
                     & EvenCycle & 1    & Infeasible & 11859360    & 10893632    & 8.610        & Infeasible & 1336188416  & 1468424277  & 979.763      & Infeasible & 92766     & 208614    & 0.135        \\
                     & EvenCycle & 2    & Infeasible & 152740084   & 167194942   & 127.611      & Infeasible & 611753448   & 711245102   & 598.378      & Infeasible & 996127    & 2340311   & 1.429        \\
                     & EvenCycle & 3    & Infeasible & 466856235   & 439993043   & 298.974      & Infeasible & 6868230379  & 7523523165  & 10800.000    & Infeasible & 2312073   & 4373462   & 2.787        \\
                     & EvenCycle & 4    & Infeasible & 54784145    & 53791828    & 36.624       & Infeasible & 285915264   & 360281393   & 505.807      & Infeasible & 2527172   & 5361369   & 3.354        \\ \cline{2-15} 
                     & Tree      & 1    & 88         & 0           & 18          & 0.000        & 88         & 0           & 18          & 0.000        & 31         & 6170      & 12737     & 0.008        \\
                     & Tree      & 2    & 61         & 0           & 18          & 0.000        & 61         & 0           & 18          & 0.000        & 25         & 1171      & 2731      & 0.002        \\
                     & Tree      & 3    & 91         & 0           & 18          & 0.000        & 91         & 0           & 18          & 0.000        & 28         & 57        & 157       & 0.000        \\
                     & Tree      & 4    & 71         & 0           & 18          & 0.000        & 71         & 0           & 18          & 0.000        & 30         & 1634      & 3433      & 0.002        \\ \hline
\multirow{12}{*}{20} & OddCycle  & 1    & Infeasible & 273165199   & 217944054   & 169.144      & Infeasible & 6147895146  & 9044193773  & 10800.000    & Infeasible & 1786324   & 3216901   & 2.101        \\
                     & OddCycle  & 2    & Infeasible & 2231047     & 1503222     & 1.467        & Infeasible & 6371122973  & 8069928711  & 10800.000    & Infeasible & 63034     & 89385     & 0.066        \\
                     & OddCycle  & 3    & Infeasible & 1049440     & 744982      & 0.542        & Infeasible & 7751043598  & 8898429750  & 10800.000    & Infeasible & 31642     & 48916     & 0.035        \\
                     & OddCycle  & 4    & Infeasible & 414762      & 249339      & 0.170        & Infeasible & 13950045871 & 14504063426 & 10800.000    & Infeasible & 56488     & 64593     & 0.051        \\ \cline{2-15} 
                     & EvenCycle & 1    & Infeasible & 7618735112  & 8591239945  & 7112.081     & \$ -       & -           & -           & Running      & Infeasible & 20719998  & 47209517  & 29.050       \\
                     & EvenCycle & 2    & 141        & 0           & 20          & 0.000        & 141        & 0           & 20          & 0.000        & 36         & 13158     & 25609     & 0.017        \\
                     & EvenCycle & 3    & Infeasible & 20754606    & 17374721    & 14.777       & Infeasible & 7684808926  & 7634511236  & 10800.000    & Infeasible & 167148    & 351029    & 0.225        \\
                     & EvenCycle & 4    & Infeasible & 15355600960 & 14057177765 & 10800.000    & Infeasible & 5426010704  & 7820985965  & 10800.000    & Infeasible & 335043320 & 686615646 & 428.222      \\ \cline{2-15} 
                     & Tree      & 1    & 64         & 0           & 20          & 0.000        & 64         & 0           & 20          & 0.000        & 24         & 107586    & 201000    & 0.132        \\
                     & Tree      & 2    & 103        & 0           & 20          & 0.000        & 103        & 0           & 20          & 0.000        & 29         & 560       & 1443      & 0.001        \\
                     & Tree      & 3    & 96         & 0           & 20          & 0.000        & 96         & 0           & 20          & 0.000        & 36         & 2908      & 6536      & 0.004        \\
                     & Tree      & 4    & 115        & 0           & 20          & 0.000        & 115        & 0           & 20          & 0.000        & 27         & 57601     & 151814    & 0.095        \\ \hline

\end{tabular}
}
\end{table}

\revisado{In order to use some of the random graphs in experiments with the BPB algorithms, we selected four
graphs of each type (with even cycle, with odd cycle and trees) for each number of vertices and
randomly weighted the edges with a uniform distribution in the interval $[1, 30]$.
We made two subsets of experiments: the first one involved only the EQ-CDGP and EQ-CDGP-Const models,
in order to find a feasible solution for each of its instances that were generated (that is, the algorithms
use the pruning procedure, but not bounding - equivalently, stopping the search as soon as the first feasible solution is found), and the second one involved the optimization models for
each discussed model (MinEQ-CDGP, MinEQ-CDGP-Const, MinGEQ-CDGP and MinGEQ-CDGP-Const).}

\revisado{Tables \ref{tbl:BP-EQ-4to8}, \ref{tbl:BP-EQ-9to16} and \ref{tbl:BP-EQ-18to20} give detailed results
with 4 to 8, 9 to 16 and 18 to 20 vertices, respectively, considering each BPB algorithm applied to the decision versions. We can see that
BPB-Prev reaches a feasible solution faster than the other methods (that is, it solves the decision problem
in less time), but it also returns the first feasible solution with a worst span than BPB-Select (noting that BPB-Prev-FeasCheckFull always returns the same span from BPB-Prev because only the pruning point is changed). However, it is much slower to prove that an instance is infeasible (that is, the answer to the decision problem is NO). This is explained by the fact that the enumeration tree is smaller in BPB-Select, since
instead of two color possibilities for each vertex, there is only one. Although the time complexity of 
determining the color for a vertex in BPB-Select is higher (as shown in Table \ref{tbl:bpb}), this is compensated
by generating a smaller tree and that the feasibility check is not explicitly needed, since it is guaranteed
by the color selection algorithm. We also note that BPB-Prev-FeasCheckFull has the worst CPU times for
infeasible instances, because the method will keep branching the enumeration tree to find a feasible solution,
but since feasibility checking is only done at the leaf nodes, the tree will tend to become the full
enumeration tree.}

\revisado{In the same manner, Table \ref{tbl:mineq-opt} shows the results for the BPB algorithms considering
the optimization versions and applied only to feasible MinEQ-CDGP instances. For almost all of these instances,
BPB-Prev is the best method, BPB-Prev-FeasCheckFull shows similar performance and BPB-Select has
worse CPU times. The ties between BPB-Prev and BPB-Prev-FeasCheckFull are explained by the fact that although, in
the latter method, the feasibility checking at the leaf nodes increases the time required to find a feasible solution, we
keep using the bounding procedures at each node, which reduces the amount of work needed to find the optimal solution.
We also note that, for the 4th Tree instance with 20 vertices, BPB-Select does not find the true optimum
for the problem. This happens because the method is recursively applied only to vertices which have recorded
neighbors, in the same manner as the other two BPBs, but the system of absolute value expressions returns only
one color, which may not be the one for the optimal solution when applying the recursion only on some vertices. On some
experiments, we detected that, if we generate all vertice orders and apply the color selection of BPB-Select, 
on them, the optimal solution is found, but the CPU times become very high, since this procedure does not take
advantage of the 0-sphere intersection characteristic.}

\revisado{The last experiments were made by applying the BPB algorithms on all instances considered for the
algorithms, but by transforming them to MinGEQ-CDGP (changing the $=$ in constraints to $\geq$). Since
they are always feasible for MinGEQ-CDGP because of its equivalence to BCP, we only have to consider optimization
problems, as was already done in Section \ref{sec:probStat}. The same pattern of previous experiments occur here,
with BPB-Prev being the best method of the three, however, the CPU time difference between it and 
BPB-Prev-FeasCheckFull becomes much more apparent here, since there are many more feasible solutions for MinGEQ-CDGP 
than MinEQ-CDGP.}

% \usepackage{multirow}
% \usepackage{graphicx}
\begin{table}[]
\centering
\caption{Results for BPB algorithms (optimization) applied to MinEQ-CDGP instances - 4 to 20 vertices.}
\label{tbl:mineq-opt}
\resizebox{\textwidth}{!}{%
\begin{tabular}{ccc|cccccC{1cm}C{1.3cm}|cccccC{1cm}C{1.3cm}|cccccC{1cm}C{1.3cm}}
\cline{4-24}
                    &               &               & \multicolumn{7}{c}{\textbf{BPB-Prev}}                                                                                                                                 & \multicolumn{7}{|c}{\textbf{BPB-Prev-FeasCheckFull}}                                                                                                                   & \multicolumn{7}{|c}{\textbf{BPB-Select}}                                                                                                                               \\ \hline
\textbf{\boldmath$|$V$|$}        & \textbf{Type} & \textbf{Inst} & \textbf{Span} & \textbf{\# Bounds} & \textbf{\# Prunes} & \textbf{\# Sol.} & \textbf{\# Nodes} & \textbf{First Time (s)} & \textbf{Total Time (s)} & \textbf{Span} & \textbf{\# Bounds} & \textbf{\# Prn.} & \textbf{\# Sol.} & \textbf{\# Nodes} & \textbf{First Time (s)} & \textbf{Total Time (s)} & \textbf{Span} & \textbf{\# Bounds} & \textbf{\# Prunes} & \textbf{\# Sol.} & \textbf{\# Nodes} & \textbf{First Time (s)} & \textbf{Total Time (s)} \\ \hline
\multirow{6}{*}{4}  & EvenCycle     & 3             & 19            & 12                 & 0                  & 3                     & 26                & 0.000                           & 0.000                         & 19            & 12                 & 0                  & 3                     & 26                & 0.000                           & 0.000                         & 19            & 8                  & 0                  & 3                     & 25                & 0.000                           & 0.000                         \\
\hline
\end{tabular}
}
\end{table}

% \usepackage{multirow}
% \usepackage{graphicx}
\begin{table}[]
\centering
\caption{Results for BPB algorithms (optimization) applied to MinGEQ-CDGP instances - 4 to 10 vertices.}
\label{tbl:mingeq-opt-4to10}
\resizebox{\textwidth}{!}{%
\begin{tabular}{ccc|cccccC{1.5cm}C{1.7cm}|cccccC{1.5cm}C{1.7cm}|cccccC{1.5cm}C{1.7cm}}
\cline{4-24}
                    &               &               & \multicolumn{7}{c}{\textbf{BPB-Prev}}                                                                                                                                 & \multicolumn{7}{|c}{\textbf{BPB-Prev-FeasCheckFull}}                                                                                                                   & \multicolumn{7}{|c}{\textbf{BPB-Select}}                                                                                                                               \\ \hline
\textbf{\boldmath$|$V$|$}        & \textbf{Type} & \textbf{Inst} & \textbf{Span} & \textbf{\# Bounds} & \textbf{\# Prunes} & \textbf{\# Sol.} & \textbf{\# Nodes} & \textbf{First Time (s)} & \textbf{Total Time (s)} & \textbf{Span} & \textbf{\# Bounds} & \textbf{\# Prn.} & \textbf{\# Sol.} & \textbf{\# Nodes} & \textbf{Time 1st (s)} & \textbf{CPU Time (s)} & \textbf{Span} & \textbf{\# Bounds} & \textbf{\# Prunes} & \textbf{\# Sol.} & \textbf{\# Nodes} & \textbf{Time to 1st (s)} & \textbf{CPU Time (s)} \\ \hline
\multirow{12}{*}{4}  & OddCycle      & 1             & 34            & 12                 & 16                 & 2                     & 41                & 0.000                   & 0.000                   &               & 12                 & 16                 & 2                     & 47                & 0.000                   & 0.000                   & 34            & 2                  & 12                 & 2                     & 40                & 0.000                   & 0.000                   \\
\hline
\end{tabular}
}
\end{table}                    

\begin{table}[]
\centering
\caption{Results for BPB algorithms (optimization) applied to MinEQ-CDGP instances - 12 to 20 vertices.}
\label{tbl:mingeq-opt-12to20}
\resizebox{\textwidth}{!}{%
\begin{tabular}{ccc|cccccC{1.5cm}C{1.7cm}|cccccC{1.5cm}C{1.7cm}|cccccC{1.5cm}C{1.7cm}}
\cline{4-24}
                    &               &               & \multicolumn{7}{c}{\textbf{BPB-Prev}}                                                                                                                                 & \multicolumn{7}{|c}{\textbf{BPB-Prev-FeasCheckFull}}                                                                                                                   & \multicolumn{7}{|c}{\textbf{BPB-Select}}                                                                                                                               \\ \hline
\textbf{\boldmath$|$V$|$}        & \textbf{Type} & \textbf{Inst} & \textbf{Span} & \textbf{\# Bounds} & \textbf{\# Prunes} & \textbf{\# Sol.} & \textbf{\# Nodes} & \textbf{First Time (s)} & \textbf{Total Time (s)} & \textbf{Span} & \textbf{\# Bounds} & \textbf{\# Prn.} & \textbf{\# Sol.} & \textbf{\# Nodes} & \textbf{Time 1st (s)} & \textbf{CPU Time (s)} & \textbf{Span} & \textbf{\# Bounds} & \textbf{\# Prunes} & \textbf{\# Sol.} & \textbf{\# Nodes} & \textbf{Time to 1st (s)} & \textbf{CPU Time (s)} \\ \hline
\multirow{12}{*}{18} & OddCycle      & 1             & 24            & 128682             & 169800             & 36                    & 303312            & 0.000                   & 0.331                   & 24            & 57939861           & 55723668           & 36                    & 147688713         & 0.011                   & 181.333                 & 23            & 89251              & 2157123            & 6                     & 4215186           & 0.000                   & 3.294                   \\
\hline
\end{tabular}
}
\end{table}

\section{Concluding remarks}
\label{sec:conc}

In this paper, we proposed some distance geometry models for graph coloring problems with distance constraints that can be applied to important, modern real world applications, such as in telecommunications for channel assignment in wireless networks. In these proposed coloring distance geometry problems (CDGPs), the vertices of the graph are embedded on the real line and the coloring is treated as an assignment of natural numbers to the vertices, while the distances correspond to line segments, whose objective is to determine a feasible intersection of them. 

We tackle such problems under the graph theory approach, to establish conditions of feasibility through behavior analysis of the problems in certain classes of graphs, identifying prohibited structures for which the occurrence indicates that it can not admit a valid solution, as well as identifying classes graphs that always admit valid solution.

We also developed exact and approximate enumeration algorithms, based on the Branch-and-Prune (BP) algorithm proposed for the Discretizable Molecular Distance Geometry Problem (DMDGP), combining concepts from constraint propagation and optimization techniques, resulting in Branch-Prune-and-Bound algorithms (BPB), that handle the set of distances in different ways in order to get feasible and optimal solutions. The computational experiments involved equality and inequality constraints and both uniform and arbitrary sets of distances, where we measure the number of prunes and bounds and the CPU time needed to reach the best
solution.

The main contribution of this paper consists of a distance geometry approach for special cases of T-coloring problems with distance constraints, involving a study of graph classes for which some of these distance coloring problems are unfeasible, and branch-prune-and-bound algorithms, combining concepts from the branch-and-bound method and constraint propagation, for the considered problems.

Ongoing and future works include improving the BPB formulations by domain reduction and more specific constraints; developing hybrid methods, involving integer and constraint programming; and applying heuristics, in order to solve the proposed distance coloring models more efficiently. Studying problems posed to specific classes of graphs, in order to establish other characterizations of feasibility conditions for more general CDGP problems, is also a subject of research in progress. \\

\vspace{0.5cm}

\noindent \textbf{\Large{Acknowledgement}}

We would like to thank Dr. Leo Liberti, professor at LIX - École Polytechnique, France, for his valuable contribution to this work, \revisado{having} identified \revisado{a} relationship between coloring problems in graphs with additional constraints on the edges and distance geometry problems, and for \revisado{discussing some
of the ideas presented here}.

\bibliographystyle{plainnat}
\bibliography{ITOR-DGapproachSpecialGraphColoring-deFreitasDiasMaculanSzwarcfiter}
\end{document}